\documentclass[12pt,authoryear]{article}

\usepackage{graphicx,subfigure}
\usepackage{amsmath}
\usepackage[title]{appendix}
\usepackage{setspace}
\usepackage{url}
\usepackage{chngcntr}
\usepackage{soul}
\usepackage{titlesec}
\usepackage[sort&compress,comma,authoryear]{natbib}

\usepackage{amsmath,amssymb,amsthm}
\usepackage[usenames,dvipsnames]{xcolor}
\usepackage{verbatim}
\usepackage{caption}
\usepackage[flushleft]{threeparttable}
\usepackage{booktabs}
\usepackage{hyperref}
\usepackage{xcolor}
\usepackage{tabularx}
\usepackage{geometry}
\usepackage{pdfpages}

\graphicspath{{../figures/}}
\usepackage{hyperref}
\usepackage{epstopdf}
\usepackage{bm}
\usepackage{enumerate}

\hypersetup{
	colorlinks,
	linkcolor={red!50!black},
	citecolor={blue!50!black},
	urlcolor={blue!80!black}
}

\usepackage{verbatim}	

\newcommand\Tstrut{\rule{0pt}{2.6ex}}         
\newcommand\Bstrut{\rule[-0.9ex]{0pt}{0pt}}   

\newcommand{\newapprox}{{\raise.17ex\hbox{$\scriptstyle\mathtt{\sim}$}}}

\newtheorem{definition}{Definition}[section]

\newtheorem{lemma}{Lemma}[section]

\newtheorem{proposition}{Proposition}[section]
\newtheorem{remark}{Remark}[section]

\title{Investor Experiences and International Capital Flows\thanks{	We thank Vladimir Asriyan, Fernando Broner, Nicolas Courdacier, Jaume Ventura, and seminar participants at CREi-UPF and the NBER ISOM conference 2019 in London for their comments and suggestions, and Clint Hamilton for excellent research assistance.}}

\author{Ulrike Malmendier \thanks{UC\ Berkeley, NBER, and CEPR. Corresponding author: Department of Economics and Haas School of Business, University of California, 501 Evans Hall, Berkeley, CA 94720-3880; ulrike@berkeley.edu.} \and  Demian Pouzo \thanks{Department of Economics, University of California, 501 Evans Hall, Berkeley, CA 94720-3880; dpouzo@berkeley.edu.} \and Victoria Vanasco  \thanks{CREI and BGSE, Ramon Trias Fargas 25-27, Barcelona, 08005, Spain; vvanasco@crei.cat.} }

\begin{document}
	\maketitle
	
	\begin{abstract}
		We propose a novel explanation for 
		classic international macro puzzles regarding capital flows and portfolio investment, which builds on modern macro-finance models of experience-based belief formation.
		Individual experiences of past macroeconomic outcomes have been shown to exert a long-lasting influence on beliefs about future realizations, and to explain domestic stock-market investment.
		We argue that experience effects can explain the tendency of investors to hold an over proportional fraction of their equity wealth in domestic stocks (home bias), to  invest in domestic equity markets in periods of domestic crises (retrenchment), and to withdraw capital from foreign equity markets in periods of foreign crises (fickleness).  Experience-based learning generates additional implications regarding the strength of these puzzles in times of higher or lower economic activity and depending on the demographic composition of market participants. We test and confirm these predictions in the data.
	\end{abstract}

\doublespacing
\newpage

\section{Introduction}

\pagenumbering{arabic}


At least since \cite{obstfeld2000six}, researchers have aimed 
to find a unifying explanation for the major puzzles in international macro-economics that is both realistic and empirically convincing. Traditionally, explanations the home bias in stock-market investments and other capital-market anomalies have largely relied on the dichotomy of traded versus non-traded goods, on trade costs, and on other frictions. The more recent literature has aimed to improve the realism of these models, which often seemed too rigid or miscalibrated in terms of magnitude \citep{coeurdacier2013home}. 

In this paper, we aim to improve a different aspect of modeling realism -- the psychological realism of investors' belief formation.
We show that the modern macro-finance models of belief formation, and in particular the notion of \emph{experience-based learning}, contribute to a better understanding of international capital flows and portfolio investments.

What is experience-based learning? Motivated by the investor response to macro-shocks such as the Great Depression or the 2008 Financial Crisis, modern models of belief formation have focused on two elements. The first is the overweighting of recent realizations. 
When predicting future returns in the stock, housing, or other asset markets, investors appear to overly rely on their observations of those markets in recent months or years.\footnote{This first element features in several of the more recent models, including natural expectation formation \citep{FusterHebertLaibson2011, FusterLaibsonMendel2010} and over-extrapolation \citep{barberis2015x, barberis2016extrapolation}.}
The second key element is the long-lasting effect of crisis experiences and the associated systematic cross-sectional differences. As conveyed by the notion of ``depression babies'' or the ``deep scars'' of the 2008 financial crisis  \citep{Blanchard2012,Malmendier_Shen2015}, macro-economic shocks alter investment and consumption behavior for decades to come. Moreover, younger cohorts tend to react significantly more strongly than older cohorts. For example, investors who have lived through financial crises tend to shy away from stock-market investment in the long-run and are pessimistic about future stock market returns, and this long-term influence is more pronounced among younger generations \citep{Malmendier_Nagel2008}.\footnote{\hspace{0mm} The same holds for IPO experiences and future IPO investment \citep{Kaustia_Knuepfer2008, HirshleiferEtAl2011}.}

Models of \emph{experience effects} naturally generate both the overweighting of recent experiences and the long-lasting effects of lifetime experiences. Experience-based learners assign extra weight to realizations of macro-financial variables that they have personally experienced when they form beliefs about future outcomes of the same variables.
A given crisis experience exerts stronger influence on younger cohorts, for whom the crisis experience constitutes a larger portion of their lifetime histories so far.\footnote{See \cite{malmendierpouzovanasco2018} and \cite{collin2016asset}, and for the long-lasting effects the model of \cite{schraeder2015information}.
}
In addition, models of experience-based learning also rationalize classical asset-pricing puzzles such as return predictability \citep{CampbellShiller1988,fama1988dividend} and excess volatility \citep{LeRoyPorter1981, Shiller1981, LeRoy2005}, and micro-level stylized facts such as investors chasing past performances.\footnote{A discussion of the broader literature on learning and asset pricing puzzles, extrapolation, departures from rational beliefs in asset pricing, and learning in OLG models is in \cite{malmendierpouzovanasco2018}.}

These insights have direct implications for international macro models, as investors in different countries have different ``experiences." They differ in their exposures to domestic versus foreign outcomes, and different countries also have different demographics. For example, applied to the Asian financial crisis, experience-based learning would predict that the crisis has exerted a significant influence on the risk attitudes and beliefs of East Asians, and its influence was different for, say, Europeans -- even controlling for income, wealth, and other standard factors. Moreover, comparing regions across and within East Asia, the influence is predicted to differ depending on where a higher fraction of \emph{young} market participants were exposed in each country.
Nevertheless, most of the international finance models to date are formulated under the rational expectation paradigm. By and large, they work under the assumption that economic agents have correct beliefs about the laws of complex economic processes.\footnote{Some exceptions include \cite{gourinchas2004exchange}, \cite{veldkamp2009information}, \cite{dziuda2012asymmetric}, \cite{ordonez2016aggregate}.}




This paper introduces the notion of \emph{experience-based learning}  (EBL) into the international macro context and shows its potential to jointly explain some of the long-standing puzzles on capital flows and portfolio investment: home bias, fickleness, and retrenchment.
Experience-based learning not only provides an alternative, psychologically and micro founded explanation 
but
also generates additional predictions that relate these puzzles to the demographic composition and size of risky-asset markets and the cyclicality of capital flows.

The pattern of \emph{home bias} in portfolio holdings was first discussed by \cite{cooper1986costs} and \cite{french1991investor}, the latter of whom calculated that, at the time, both Americans and Japanese investors held more than 90\% of their equity wealth in respective home countries' stocks.
The bias remains strong across countries today, though it has declined relative to the late 1990s and is smaller in smaller countries \citep{cooper2013equity}. \emph{Retrenchment} describes the pattern of domestic capital inflows increasing during periods of domestic or global crisis \citep{forbesetal2012, broner2013gross}, and \emph{fickleness} describes the pattern of foreign capital outflows increasing during periods of domestic or global crises \citep{forbesetal2012, caballero2016model}. 
Our model jointly rationalizes these puzzles and generates new testable predictions that account for the demographic composition of different countries
and the interaction with cycles of economic booms and busts.

Our model set-up extends the CARA-Normal OLG framework of \cite{malmendierpouzovanasco2018} to a two-country setting. Agents maximize their end-of-life consumption, and during their lifetime can invest in (i) the domestic and the foreign risky assets, i.e. a claim to the stream of future outputs of the respective country, each in unit supply, and (ii) a risk-free asset in fully elastic supply that pays return $R$.
We assume that all agents know the distribution of outputs for both countries, including their variance, but that they are uncertain and learn about the output means.

Experience-based learning means that agents overweight realizations observed during their lifetimes when forecasting output. That is, even though all agents observe the entire history of output realizations of both countries, they choose to weigh observations more if they have personally \emph{experienced} them. Furthermore, to capture that agents are more confident about their knowledge of their own country than of a foreign country, we assume that prior beliefs about domestic output are more precise than prior beliefs about the foreign output. Note that, as all information is available to all agents, our model differs from models of asymmetric information, and prices and portfolios do not contain relevant information about future output realizations. 

Turning to the equilibrium analysis, we first consider the benchmark case where countries are symmetric, i.\,e., all agents have the same prior beliefs and both countries have the same demographics. In this benchmark scenario, both countries hold the same (aggregate) portfolio and asset prices are constant over time.\footnote{We abstract away from the possibility of rational bubbles, which may introduce dynamics to asset prices that are unrelated to fundamentals.} This is our version of the \emph{mutual-fund theory} \citep{markowitz1959portfolio}, which states that all agents should hold the efficient (market) portfolio. This result is important in that it shows that EBL alone does not generate heterogeneity in aggregate portfolio across countries. It does, however, generate heterogeneity in portfolio holdings across cohorts within a country, as different generations have different experiences \citetext{cf.\;\citealp{malmendierpouzovanasco2018}}.

This prediction changes when we allow EBL agents to hold more precise prior beliefs about their own country's output, when countries differ in the demographics of their market participants, or both.
We begin by characterizing the equilibrium when prior beliefs about domestic and foreign output differ, but both countries have the same demographics. In that case, the equilibrium price of a country's risky asset varies with past realizations of that country's output.\footnote{This result relies on the assumption that output realizations are independent across countries; else, prices would depend on all output realizations. We discuss the role of output correlation in Section \ref{sec: Equilibrium}.} Specifically, it depends on the history of past output realizations observed by the oldest market participant. As a result, demands and resulting portfolio holdings also depend on such history, in line with the findings in \cite{malmendierpouzovanasco2018}.

%
%
This setting also generates home bias 
in portfolio holdings as agents perceive domestic output as less risky.
Moreover, agents over-react to foreign output realizations, generating the observed pattern of retrenchment and fickleness of capital flows in response to shocks: 
As agents are less confident about what they know about the foreign country, personal experiences of foreign-output realizations strongly influence their respective posterior beliefs. As a result, they over-react to negative shocks, e.\,g., low output realizations abroad, which, together with general equilibrium effects, imply that capital inflows of domestic agents (retrenchment) and outflows of foreign agents (fickleness) both increase -- even though all agents perceive the shock as negative. The reverse holds for booms, which are followed by the corresponding decreases.
The effect is observed in both countries during global shocks.
These predictions are consistent with, and provide an explanation for, the evidence in \cite{broner2006peril,milesi2011great,broner2013gross}, who find that domestic inflows and foreign outflows are counter-cyclical, and that retrenchments and fickleness occur during both domestic and global crises. They are also consistent with the negative correlation of (global and domestic) growth proxies with domestic inflows (``retrenchment") and foreign outflows (``stops") documented in \cite{forbesetal2012}.

Next we extend the model to consider the role of cross-country differences in the demographics of market participants.
We analyze how the reaction of capital flows to a given shock varies with the demographic composition of market participants in each country. We find that when a country has a larger number of young market participants, it over-reacts to both domestic and foreign shocks relative to the baseline model. As a result, the retrenchment and fickleness effects are alleviated in this country. Vice versa, retrenchment and fickleness are exacerbated in countries with a smaller number of young market participants.

Finally, we take our model predictions to the data. We use data from the IMF, World Bank, and the World Federation of Exchanges. First, we confirm positive US\;home bias in every period of our sample from 1980 to 2017. We then turn to a panel of G20 countries with a sample period from 1971-2017 and provide evidence of fickleness (outflow of foreign funds) and retrenchment (inflow of domestic funds) after recessions. Turning to the novel predictions of our model, we use UN population data to test for differences by demographics. We show that both capital-flow patterns, fickleness and retrenchment, are exacerbated when the number of older people in a given country is particularly high. Through the lenses of our model, the data also suggests that longer experiences (older age) is empirically more predictive of precise priors about a country than being a domestic (rather than foreign) investor, underscoring the importance of experience effects.

Overall our findings illustrate that modern approaches to belief formation, which acknowledge the long-lasting effects of prior experiences in domestic and foreign markets, help explain longstanding puzzles in international capital flows and generate further testable predictions that hold up in the data.

\bigskip

\noindent\textbf{Related Literature.}
The literature on \emph{experience effects} builds on the seminal work by Kahneman and Tversky on availability bias. \cite{TverskyKahneman1974} show that, when individuals form beliefs about future realizations of stochastic variables, they tend to assign extra weight to information that is easily ``available'' to them. Such information tends to be personally experienced outcomes, and in particular recent realizations. Building on this insight, \cite{weber1993determinants} provide extensive evidence that ``learning from experience'' is significantly more powerful than ``learning from description.''\footnote{\hspace{0mm} Cf.\;also \cite{Hertwigetal2004} and \cite{simonsohn2008tree}.}

In the economics literature, a growing body of evidence documents experience effects in economic decision-making. In addition to the theoretical and empirical work cited above, several papers show that investors who have experienced high inflation tend to overestimate future inflation and interest rates, and thus invest more in housing (as an inflation hedge). For the same (perceived) hedging reason, they also tend to finance their housing with fixed-rate mortgages \citetext{cf.\;\citealp{Malmendier_Nagel2013, Botsch_Malmendier2018, Malmendier_Steiny2018}}.\footnote{\hspace{0mm} \cite{Malmendier_Nagel_Yan2018} show that even FOMC members' stated inflation beliefs are strongly affected by their personal lifetime experiences.} In the realm of consumption decisions, consumers who have experienced periods of economic downturn and high unemployment rates are more careful in their spending \citep{Malmendier_Shen2015}.

Turning to the international macro puzzles our analysis aims to explain, a large literature rationalizes the equity home bias, first documented by \cite{french1991investor} and \cite{tesar1995home}. \cite{coeurdacier2013home} and \cite{cooper2013equity} survey and summarize the leading explanations, including hedging motives  \citetext{e.\,g.,  \citealp{solnik1974equilibrium,fidora2007home,coeurdacier2010international,coeurdacier2016bonds}}, explicit costs and barriers to entering foreign stock markets \citetext{e.\,g., \citealp{black1974international,stulz1981effects,errunza1985international}}, information asymmetries, and alternative approaches. 

The theory proposed in this paper lies in the intersection of behavioral and informational approaches. Several papers highlight the importance of informational asymmetries between domestic and foreign investors for understanding international capital flows.\footnote{The literature relying on informational asymmetries documents the roles of distance, language, and culture in generating home bias \citep{grinblatt2001distance,malloy2005geography,bae2008local}, but also points to home bias at the local level, e.\,g., towards local companies \citep{coval2001geography}.} \cite{gehrig1993information} shows that, in a two-country, noisy rational expectations setting, home bias is a natural response to domestic agents having private information about domestic assets. Relatedly, \cite{brennan1997international} and \cite{brennan2005dynamics} model a dynamic setting with private and public information, in which domestic investors over-react to public signals about the foreign country relative to foreign agents, as foreign agents also have private information. This over-reaction to news is similar to the one described in our model, with the twist that we allow all information to be public. On a similar note, \cite{veldkamp2009information,mondria2010puzzling,dziuda2012asymmetric} argue that home bias can result from small informational advantages or financial frictions combined with agents' information processing constraints.\footnote{Asymmetric information in international finance models can also rationalize trading patterns such as return chasing, volatility of capital flows, and positive correlations between inflows and outflows \citetext{e.g. \citealp{albuquerque2009global,tille2014international,dvovrak2003gross}}.}
Here, we postulate that agents overweight observations that they experience even when they do not have informational advantages.

Relatedly, we assume that priors about domestic assets are more precise than those about foreign assets, but not due to private information. Information-based rationalizations, such as \cite{guidolin2005home}, generate predictions that differ from the ones in our model, and from empirical observations. First, in such a model, agents learn from prices and portfolio holdings, which is not the case in our environment. Thus, a theory of private information needs to explain why foreign agents do not mimic the portfolio holdings of domestic agents. Second, in such models, domestic investors should earn higher expected returns than foreigners. The evidence on this prediction is mixed \citep{coeurdacier2013home,ardalanequity}. In contrast, in our setting, a more precise prior does not imply that domestic agents earn higher expected returns, as prior beliefs need not be centered around the truth. Finally, in our model, domestic agents over-react to any information about the foreign country, generating the documented cyclicality of capital flows. These patterns cannot be easily generated with models of private information. For example, retrenchment of capital may not occur if domestic agents have more negative private information about their own country than foreign agents, which is a possibility during downturns.
That is, the theory proposed in this paper is able to jointly rationalize home bias and other documented patterns of capital flows, such as retrenchment and fickleness during global recessions, while most equity home bias models fail to rationalize these facts jointly.

Our results also relate to a large literature on global imbalances and their relation to demographics and saving decisions \citep{bernanke2005global,caballero2008equilibrium,caballero2009global,coeurdacier2015credit}. This paper contributes to this line of research by providing a novel connection between the demographic composition of market participants and the sensitivity of capital flows to shocks, i.\,e., the cyclicality of gross capital flows.

\bigskip

The remainder of the paper is organized as follows. Section \ref{sec:Baseline} presents the model setup, the notion of experience-based learning, and our equilibrium concept. Section \ref{sec:EqChar} characterizes equilibrium demand and prices for the full-information benchmark and the baseline model. Section \ref{sec:Heterogeneity} introduces heterogeneity in the demographic structures of markets across countries.
In Section \ref{sec:EmpiricalImplications}, we test the predictions of the model, and Section \ref{sec: Conclusions} concludes. All proofs are relegated to the Appendix.

\section{Model Setup}\label{sec:Baseline}

Consider two countries $H$ and $F$, each populated with overlapping generations of a continuum of risk-averse agents. Time is infinite and indexed by $t$. At each point in time $t \in \mathbb{N}_{0}$, a new generation is born in each country and lives for two periods.\footnote{We define a period as the time between time $t$ and $t+1$.} Hence, there are three generations alive at any $t$: the young, the old, and the retired. The generation born at $t=n$ is called generation $n$, and each generation has a mass of $\frac{1}{4}$ identical agents.

\emph{Preferences.} All agents have CARA preferences with risk aversion $\gamma$, and they maximize per-period utility.\footnote{\hspace{0mm} This myopia assumption simplifies the maximization problem considerably and highlights the main determinants of portfolio choice generated by experience-based learning. It is commonly used in finance (see \cite{Vives2008}).} Generation $n$ in country $i\in\{H,F\}$ is born with an endowment of 
$\mathcal{W}^{i,n}$
consumption goods. Agents can transfer resources across time by investing in financial markets, where trading takes place at the beginning of each period. At the end of the last period of their lives, agents consume the wealth they have accumulated. Hence, the young and the old generations both participate in financial markets, while the retired simply consume. Figure \ref{f: timeline} shows the timeline of this economy.

\emph{Financial Markets.}  All agents have access to a risk-free asset, which pays a gross return given by $R$, and two risky assets in unit net supply, each being a claim to the stream of future risky outputs of a country $j\in\{H,F\}$, $\{y_{j,t}\}$. That is, the dividend paid on country $j$'s risky asset corresponds to country $j$'s output. We assume outputs to be independent across countries and time and identically distributed, $y_{j,t}\sim N\left(\theta_j,\sigma_j^{2}\right)$ at time $t$ for $j\in \{H,F\}$.

Our model describes the decisions of different generations of domestic and foreign investors to invest in domestic and foreign assets. To keep notation tractable, we use superindices for the investor's country of origin $i$ and her generation $n$, and subindices for the asset's country of origin (issuance) $j$ and time $t$.

\subsection{The Agents' Problem} \label{sec:Max}

Agents from generation $n$ in country $i$ that participate in financial markets have the following budget constraint at any time $t \in \{n,n+1\}$
\begin{align}\label{eqn:BC}
\mathcal{W}^{i,n}_{t} = x^{i,n}_{H,t} \cdot p_{H,t} + x^{i,n}_{F,t} \cdot p_{F,t}  + a^{i,n}_{t},
\end{align}
where $\mathcal{W}_{t}^{i,n}$ denotes the wealth of generation $n$ of country $i$ at time $t$, $x^{i,n}_{j,t}$ and $a^{i,n}_{t}$ are the investments of generation $n$ of country $i$ in the risky asset of country $j\in\{H,F\}$ and in the risk-free asset, respectively, and $p_{j,t}$ is the price of one unit of the risky asset of country $j$ at time $t$. As a result, wealth next period is
\begin{align}\label{eqn:LoM-W}
\mathcal{W}^{i,n}_{t+1} =&\; x^{i,n}_{H,t}\cdot ( p_{H,t+1} + y_{H,t+1}) + x^{i,n}_{F,t}\cdot ( p_{F,t+1} + y_{F,t+1})  + a^{i,n}_{t} R \\
=& \sum_{j\in\{H,F\}} x^{i,n}_{j,t} ( p_{j,t+1} + y_{j, t+1} - p_{j,t}R) + \mathcal{W}^{i,n}_{t} R.
\end{align}
We denote the excess payoff received in $t+1$ from investing at time $t$ in one unit of the risky asset of country $j$, relative to the risk-free asset, as $s_{j,t+1} \equiv p_{j,t+1} + y_{j,t+1} - p_{j,t} R$.
In our framework with CARA preferences, this is analogous to the equity premium.
Using this notation, $\mathcal{W}^{i,n}_{t+1} = \sum_{j\in\{H,F\}} x^{j,n}_{i,t} s_{j,t+1} + \mathcal{W}^{i,n}_{t} R$.

To model uncertainty about fundamentals, we assume that agents do not know the mean of the output processes, $\theta_H$ and $\theta_F$, but can use past realizations to learn about them. For tractability, the variance of output is $\sigma^2$ for both countries, and is known by all agents at all times; see Section \ref{sec:learning} below for more details.

Given a wealth level $\mathcal{W}_t^{i,n}$, the problem of generation $n$
at each time $t \in \{ n,n+1 \}$ is to choose
$\{x^{i,n}_{j,t}\}$
to maximize $E_{t}^{i,n}[-\exp(-\gamma \mathcal{W}_{t+1}^{i,n})]$:
\begin{align}\label{eq: Generation n Problem}
\max_{\{x_H,x_F\} \in \mathbb{R}^2} E_{t}^{i,n}\left[-\exp(-\gamma (x_H s_{H,t+1} + x_F s_{F,t+1}))\right],
\end{align}
where $E_{t}^{i,n}\left[\cdot\right]$ is the (subjective) expectation with respect to a joint Gaussian distribution that we will define below. Note that, when $x^{i,n}_{j,t}$ is negative, generation $n$ of country $i$ is short-selling the risky asset of country $j$ at time $t$.

\begin{figure}[t]
	\centering
	\includegraphics[scale=1]{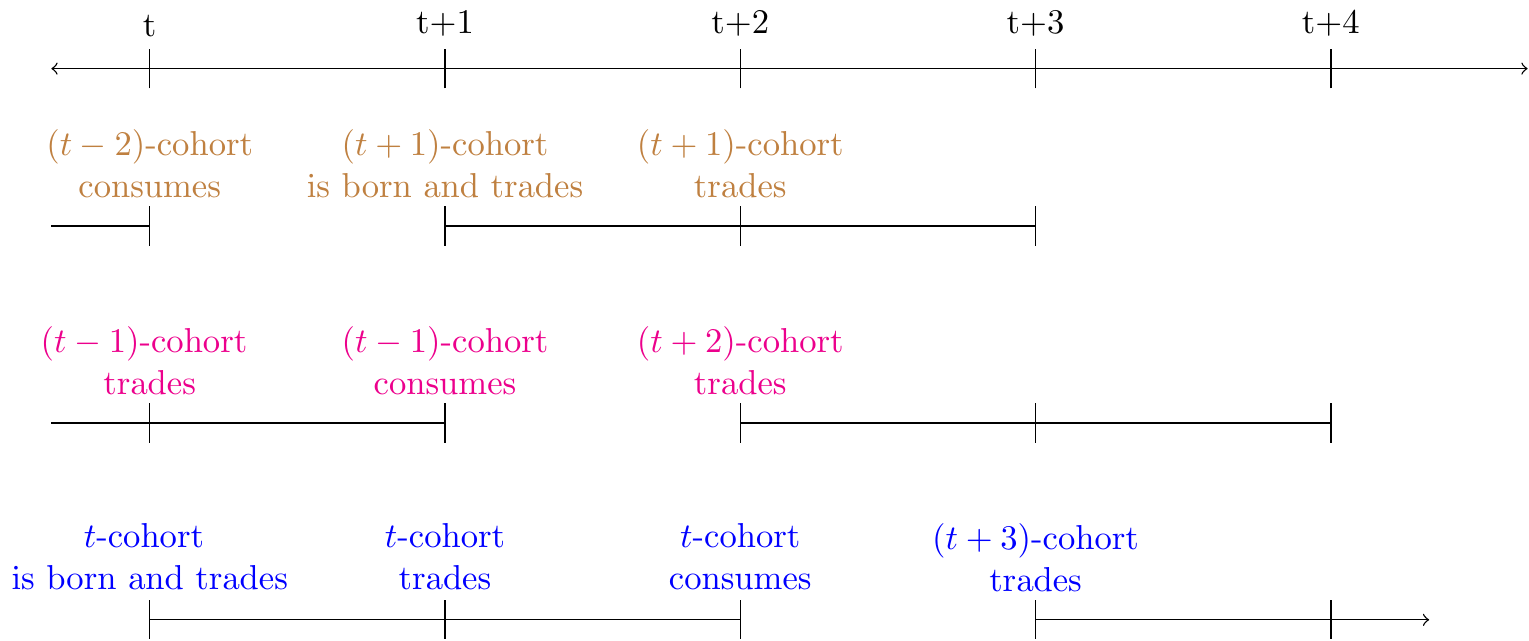}
	\caption{Timeline.} \label{f: timeline}
	\vspace{0.7cm}
\end{figure}



\subsection{The Learning Model}
\label{sec:learning}

Over time, agents learn about those aspects of the stochastic output process they are uncertain about. In order to capture experience effects and recency bias, we allow agents to depart from the standard Bayesian learning paradigm in terms of how they weight past observations when they form their beliefs.
However, we also want to retain some of the convenient features of the Bayesian framework, namely the concept of belief over subjective models. We do so by implementing the generalization of Bayesian updating in \cite{bissiri2016general}. We will show that their generalization implies that seemingly dissimilar learning procedures used in economics --- such as Bayesian and frequentist learning, and departures from these such as ``pseudo Bayesian" and experience-based learning (EBL) as in \cite{malmendierpouzovanasco2018} --- can be cast as part of an encompassing class. Thus, we think this generalization might be of independent interest for modeling departures of Bayesian learning, and we present it in more generality than needed for this paper.

We first describe agents' subjective model, i.\,e., the parametrizations of the stochastic process that the agents consider. Second, we describe the extension of Bayesian updating proposed in \cite{bissiri2016general}. We then provide a series of examples that illustrate the scope of this framework and study the learning specification used in this paper.

\subsubsection{Agents' subjective model}

Agents want to learn about the distribution of the process $(y_{t})_{t=0}^{\infty}$ where $y_{t} = (y_{F,t},y_{H,t})$ for all $t \in \mathbb{N}_{0}$. As we assume this process is iid, it suffices to specify the density function for $y_{t}$.\footnote{The iid assumption is done for simplicity; it is straightforward to extend the framework to a Markov model for output.} The subjective model is given by a family of density functions, $\{q_{\theta}\colon\theta\in\Theta\}$
where $q_{\theta}$ is a joint probability of the output vector $(y_F,y_H)$, parametrized by $\theta \in \Theta$, and $\Theta$ is the parameter set. The mapping $\theta \mapsto q_{\theta}$ is one-to-one so that each index $\theta$ fully characterizes one (and only one) density.

Agents
are uncertain only about the output process, and use data on past realizations of output to form beliefs about the densities in the subjective model that best fit the data. As typically assumed in the literature, agents correctly understand the process driving prices, i.\,e., the equilibrium mapping from output realizations to prices (see equation \eqref{e: LinearPrices} below).\footnote{Extending this assumption, while outside the scope of this paper, is not difficult. It requires considering a subjective model over the process of output and prices, where agents use past realization of all variables to update their beliefs.}

\subsubsection{Agents' Beliefs}

The objective is to model the belief of the generation $n$ at time
$t\geq n$ in country $i$, which is formally defined as a probability measure over the parameter set, i.\,e.,  $\mu_{t}^{i,n}\in\mathcal{P}(\Theta)$ where $\mathcal{P}(\Theta)$ denotes the set of probability measures over $\Theta$.

The learning problem of the generation $n$ in country $i$ is as follows. Given a prior $\mu_{0}^{i}\in\mathcal{\mathcal{P}}(\Theta)$, a scaling parameter $\varrho>0$ and, for each time $t$, a loss function, $\ell^{i,n}_{t}:\mathbb{R}^{2t}\times\Theta\rightarrow\mathbb{R}$, the posterior of generation $n$ at time $t\geq n$ in country
$i$, is
\[
\mu_{t}^{i,n}=\arg\min_{\nu\in\mathcal{P}(\Theta)}\int\ell_{t}^{i,n}(y_{0:t},\theta)\nu(d\theta)+ \varrho DIV(\nu,\mu_{0}^{i}),
\]
where $DIV(\nu,\mu)\equiv\int\log\frac{d\nu}{d\mu}(\theta)\nu(d\theta)$
is the KL divergence and $y_{0:t} = (y_{0},...,y_{t})$.\footnote{Formally $\theta\mapsto f(\theta)\equiv\frac{d\nu}{d\mu}(\theta)$
	is the derivative of $\mu$ with respect to $\nu$, i.\,e., $\int_{A}f(\theta)\mu(d\theta)=\nu(A)$
	for any Borel set $A$. It is well-defined provided $\nu$ is absolutely
	continuous with respect to $\mu$. If this is not the case, we simply
	set $DIV$ to $\infty$.}

This optimization problem can be interpreted through the lens of rational inattention (see \cite{mackowiak2018survey} for a nice review) in the sense that the agent is ``buying" a belief to maximize the payoff $-\ell^{i,n}_{t}$ and paying a ``cost" given by $-\varrho DIV(\cdot, \mu^{i}_{0})$. In fact, for $\varrho>0$, the solution is given by
\[
\mu_{t}^{i,n}(A)=\frac{\int_{A}\exp\{-\ell_{t}^{i,n}(y_{0:t},\theta)/\varrho\}\mu_{0}^{i}(d\theta)}{\int_{\Theta}\exp\{-\ell_{t}^{i,n}(y_{0:t},\theta)/\varrho\}\mu_{0}^{i}(d\theta)}
\]
for any Borel set $A\subseteq\Theta$. It is easy to see that in fact
$\mu_{t}^{i,n}\in\mathcal{P}(\Theta)$. For $\varrho=0$, $\mu_{t}^{i,n}$
puts probability one to $\arg\min_{\theta\in\Theta}\ell_{t}^{i,n}(y_{0:t},\theta)$; thus, $\varrho=0$ characterizes ``frequentist" learning scheme.
We now discuss some examples.

\subsubsection{Examples}

\textbf{Bayesian Learners. }First consider the case where the loss function is given by the log-likelihood
function, i.\,e., $\ell_{t}^{i,n}(y_{0:t},\theta)\equiv-(t+1)^{-1}\sum_{k=0}^{t}\log q_{\theta}(y_{k})$.
If $\varrho=(t+1)^{-1}$, then
\[
\mu_{t}^{i,n}(A)=\frac{\int_{A}\prod_{k=0}^{t}q_{\theta}(y_{k})\mu_{0}^{i}(d\theta)}{\int_{\Theta}\prod_{k=0}^{t}q_{\theta}(y_{k})\mu_{0}^{i}(d\theta)},~\forall~A \subseteq \Theta~Borel
\]
which is the standard Bayesian posterior for the subjective model
$\{q_{\theta}\colon\theta\in\Theta\}$ given prior $\mu_{0}^{i}$.
If $\varrho=0$, then $\mu_{t}^{i,n}$ puts probability one to the
maximum likelihood estimator, thus modeling the beliefs of a ``frequentist'' using the log-likelihood as the loss function.

\bigskip{}

\textbf{Learners from experience. }Motivated by the prior empirical evidence on experience effects, we are interested in agents who do not attach equal weight to all past observations. Under our formulation, this feature can easily be captured by simply considering a different loss function given by $\ell_{n+age}^{i,n}(y_{0:n+age},\theta)\equiv-\sum_{k=0}^{n+age}\varpi(k)\log q_{\theta}(y_{k})$, where $(\varpi(k))_{k=0}^{n+age}$ are weights chosen by the researcher and $n + age$ is the calendar time, decomposed in the generation $n$'s birth year and age.

For instance, in order to model a ``pseudo-Bayesian learner'' who
only considers observations during his lifetime (cf.\;\cite{malmendierpouzovanasco2018}), we
set $\varpi(k)=0$ for all $k<n$ and $\varpi(k)=\frac{1}{age+1}$
for all $k \in \{n,...,n+age\} $, and $\varrho=\frac{1}{age+1}$, and thus obtain
\[
\mu_{n+age}^{i,n}(A)=\frac{\int_{A}\prod_{k=n}^{n+age}q_{\theta}(y_{k})\mu_{0}^{i}(d\theta)}{\int_{\Theta}\prod_{k=n}^{n+age}q_{\theta}(y_{k})\mu_{0}^{i}(d\theta)}, ~\forall~A \subseteq \Theta~Borel.
\]

To model the case for experience-based learners who also overweight recent observations (as modeled in \cite{malmendierpouzovanasco2018}), we set $w(k)=0$
for all $k<n$ and $(\varpi(k))_{k=n}^{n+age}$ to be increasing,
e.\,g, $\varpi(k)=\frac{\beta^{n+age-k}}{\sum_{b=0}^{age}\beta^{age-b}}$
for all $k\in\{n,...,n+age\}$ where $\beta\in(0,1)$. In this case,
with $\varrho=1$, it follows that
\[
\mu_{n+age}^{i,n}(A)=\frac{\int_{A}\prod_{k=n}^{n+age}q_{\theta}(y_{k})^{\varpi(k)}\mu_{0}^{i}(d\theta)}{\int_{\Theta}\prod_{k=n}^{n+age}q_{\theta}(y_{k})^{\varpi(k)}\mu_{0}^{i}(d\theta)},~\forall~A \subseteq \Theta~Borel.
\]

\bigskip{}

\textbf{The Gaussian Case. }We denote the specification we will use throughout the rest of the paper as the ``Gaussian Case.'' For country $i \in \{H,F\}$, the subjective model is Gaussian with covariance matrix $\Sigma = \sigma^{2} I$, and the prior $\mu^{i}_{0}$ is also Gaussian, with mean $(m^{i}_{i},m^{i}_{j})$ and covariance matrix $\Sigma_{0} \equiv [1/\tau^{i}_{i},0;0,1/\tau^{i}_{j}]$ for $j \ne i$. The parameter set is $\Theta = \mathbb{R}^{2}$, i.\,e., agents only learn about the mean of the subjective model and
the two output processes are independent.

For general weights $(\varpi(k))_{k=n}^{n+age}$, the posterior of generation $n$ at time $n+age\geq n$ in country $i$ is
\[
\mu_{n+age}^{i,n}=\arg\min_{\nu\in\mathcal{P}(\Theta)}\int\sum_{k=n}^{n+age}\left\{ 0.5\varpi(k)(y_{k}-\theta)^{T}\Sigma^{-1}(y_{k}-\theta)\right\} \nu(\theta)d\theta+\varrho DIV(\nu,\mu_{0}^{i}).
\]

For $\varrho=0$, it follows that $\mu_{n+age}^{i,n}$ will assign probability
one to $\hat{\theta}^{i,n}_{n+age} = (\hat{\theta}^{i,n}_{i,n+age},\hat{\theta}^{i,n}_{j,n+age}) \equiv\sum_{k=n}^{n+age}\varpi(k)y_{k}$.
For $\varrho=1$, it follows that:\footnote{The fact that we consider $\varrho = 1$ --- as opposed to any $\varrho > 0$ --- comes with little loss of generality because what matters for the posterior belief is $\varrho \sigma^{2} $ and not $\varrho$ individually.}

\begin{proposition} \label{p: posterior_beliefs}
	The posterior belief $\mu_{n+age}^{i,n}$ of generation $n$ at time $n+age\geq n$ in country
	$i \in \{H,F\}$ is the product of two Gaussians with means
	\begin{align*}
	\hat{\theta}^{i,n}_{j,n+age} =   \frac{age+1}{age+1 + \tau^{i}_{j} \sigma^{2}} \frac{1}{age+1}\sum_{k=n}^{n+age}\varpi(k)y_{j,k}+ m^{i}_{i} \frac{\tau^{i}_{j} \sigma^{2}}{age+1 +  \tau^{i}_{j} \sigma^{2}}
	\end{align*}
	and variances \begin{align*}
	\sigma^{i,n}_{j,n+age} = \frac{\sigma^{2}}{age+1 + \tau^{i}_{j}\sigma^{2}},
	\end{align*}
	for any $j \in \{H,F\}$.
\end{proposition}


Throughout the paper, we make three additional assumptions. First, we focus on constant weights, i.\,e., $\varpi(k) = 1$ for all $k \in \{n,...,n+age\}$. Generalizing this assumption to allow for more flexible weight profiles --- for instance, profiles that model recency bias --- will be interesting when considering longer-lived agents.

Second, we impose symmetry in the precisions, in the sense that $\tau \equiv \tau^{i}_{i}$ and $\tau^{\ast} \equiv \tau^{j}_{i}$ for all $j \ne i$ and any $i \in \{H,F\}$.
(Here and in other instances, we use a star $^*$ to indicate ``foreign country" when the context does not require keeping track of all the information in the indices, allowing us to highlight cross-country symmetries.)
Given Proposition \ref{p: posterior_beliefs}, this assumption readily implies symmetry in the posterior standard deviation, i.\,e., $\sigma_{age} \equiv \sigma^{i,n}_{i,n+age}$ and $\sigma^{\ast}_{age} \equiv  \sigma^{j,n}_{i,n+age}$ for any $j \ne i$ and any $i \in \{H,F\}$. In addition, this assumption also implies symmetry on the weights of the posterior mean, i.\,e., for any $i \in \{H,F\}$ and any $j \ne i$,
	\begin{align} \label{eq: posterior_beliefs}
\hat{\theta}^{i,n}_{i,n+age} =  w_{age} \sum_{k=n}^{n+age}\omega(age) y_{j,k}+ (1-w_{age}) m^{i}_{i} \\
\hat{\theta}^{i,n}_{j,n+age} =  w^{\ast}_{age} \sum_{k=n}^{n+age}\omega(age) y_{j,k}+ (1-w^{\ast}_{age}) m^{i}_{j}
\end{align}
where $w_{age} \equiv  \frac{age+1}{age+1 + \tau \sigma^{2}}$, $w^{\ast}_{age} \equiv  \frac{age+1}{age+1 + \tau^{\ast} \sigma^{2}}$ and $\omega(age) = 1/(age+1)$.

Third, we assume that domestic agents have more precise prior beliefs about their own country's output than about the foreign country's output, i.\,e., $\tau^{\ast} < \tau$. This assumption captures the idea that agents feel more confident about their knowledge of their own country, which they have experienced since childhood, than about their knowledge of the foreign country. This could be due to over-confidence, to early life experiences, or to some form of inter-generational information transmission. As agents have more precise priors about their own country's output, and all agents observe the same information, agents will also have more precise posteriors about their own country's output. As a result, the posterior standard deviation for foreign output will be higher than for domestic output, $\sigma_{age} < \sigma_{age}^*$, for any $age$ (i.\,e., cohort). In addition, old agents have more ``experience" than younger agents, and thus have a more precise posterior belief for a given country's output: $\sigma_{age} > \sigma_{age+1}$ and $\sigma^{\ast}_{age} > \sigma^{\ast}_{age+1}$.

We conclude this part by pointing out that, under experience-based learning, the prior has a non-vanishing role on the formation of beliefs. Moreover, the weight that an agent assigns to the output realization $y_{j,t}$ for $j\in\{H,F\}$ when forming her beliefs varies with her $age\in\{0,1\}$ and her country of origin $i\in \{H,F\}$. These features, which are not present in the standard Bayesian learning, are key to generate heterogeneity on the response to macroeconomics shocks of different generations in different countries, and as such, are key drivers of the results in the paper.

%

\subsection{Equilibrium Definition} \label{sec: Equilibrium}

We now proceed to define the equilibrium of
the economy.

\begin{definition}[Equilibrium]
	An equilibrium is a demand profile for the risky assets
	$\{x^{i,n}_{j,t}\}$,
	a demand profile for the risk-free asset
	$\{a^{i,n}_{t}\}$, and price schedules
	$\{ p_{j,t}\}$ such that
	\begin{enumerate}
		\item  given the price schedule, $\{ (x^{i,n}_{H,t},x^{i,n}_{F,t},a^{i,n}_{t} ) :
		t \in \{ n,n+1 \}\}$ solve the maximization problem \eqref{eq: Generation n Problem} of generation $n$ in country $i\in\{H,F\}$;
		\item markets clear in all $t \in \mathbb{N}_{0}$: \begin{align}
		1 =& \frac{1}{4}\left(  x^{H,t}_{H,t} + x^{H,t-1}_{H,t} + x^{F,t}_{H,t} + x^{F,t-1}_{H,t} \right), \\
		1 =& \frac{1}{4}\left(x^{H,t}_{F,t} + x^{H,t-1}_{F,t} +  x^{F,t}_{F,t} +  x^{F,t-1}_{F,t} \right).
		\end{align}
	\end{enumerate}
\end{definition}

We denote country $i$'s aggregate demand of the asset of country $j$ as
\begin{equation} \label{eq: aggregate_demands}
X^i_{j,t} \equiv \frac{1}{4}\left(x^{i,t}_{j,t} +  x^{i,t-1}_{j,t}\right)\end{equation}
the portfolio of risky assets of country $i$ at time $t$ as $\bar{X}_t^i \equiv \{X^i_{H,t},X^i_{F,t}\}$.
Finally, we focus the analysis on the class of linear equilibria, i.\,e., equilibria with affine prices:

\begin{definition}[Linear Equilibrium]\label{d: linear_eq}
	A linear equilibrium is an equilibrium wherein prices are an affine function of output realizations. That is,  there exists a $K \in \mathbb{N}$, $\alpha_j \in \mathbb{R}$, $\beta_{j,k},\beta^*_{i,k} \in \mathbb{R}$ for all $k \in \{0,...,K\}$,  such that
	\begin{equation} \label{e: LinearPrices}
	p_{j,t}=\alpha_j +\sum_{k=0}^{K}\beta_{j,k} \cdot y_{j,t-k} + \sum_{k=0}^{K}\beta^*_{i,k} \cdot y_{i,t-k}.
	\end{equation}
	for all $j,i\in\{H,F\}$ with $i\neq j$.
\end{definition}

\section{Equilibrium Characterization}\label{sec:EqChar}

\subsection{Full-Information Benchmark}

Before we analyze the equilibrium portfolio choices under learning, we derive the demand for risky assets in the benchmark case of known mean outputs $\theta_H$ and $\theta_F$.
In this scenario, there are no disagreements across cohorts nor across countries. The demands of any cohort trading at time $t$ solve \eqref{eq: Generation n Problem}.
The solution to this problem is standard, with \begin{align}\label{eq: RE_Problem}
x^{i,n}_{j,t} = \frac{E^{i,n}_{t}\left[s_{j,t+1}\right]}{\gamma V_t^{i,n}[s_{j,t+1}]}
\end{align}
for all $n\in\{t-1,t\}$, and zero otherwise. As all agents share the same posterior beliefs $E_t^{i,n}[y_{j,t+1}]=\theta_j$ for all $i,j,n,t$, there is no heterogeneity in cohorts' demands. As each risky asset is in unit supply, 
market clearing implies
$x_{j,t}^{i,n} = 1$
for all $n\in\{t-1,t\}$, and zero otherwise. Furthermore, there exists a unique equilibrium with prices $p_{j} =
\frac{\theta_j - \gamma \sigma^2}{R-1}$
for all $t$.\footnote{\hspace{0mm} Our analysis focuses on Fundamental Equilibria, as we rule out the presence of price bubbles.}

\subsection{EBL: Symmetric Countries Benchmark}

We now return to EBL agents, and first characterize portfolio choice in an economy with no heterogeneity across countries. That is, agents do not know output means and are learning through EBL, but they share prior beliefs about all outputs. The main results are summarized in the following Lemma.

\begin{lemma} \label{l: symmetry}
	If all agents are born with the same prior belief about domestic and foreign output, then all countries hold the same fraction of the world portfolio. As a result, aggregate holdings of domestic and foreign assets are \begin{align}
	X_{H,t}^{H}=X_{H,t}^{F}=X_{F,t}^{H}=X_{F,t}^{F}=1,
	\end{align}
	and thus, do not vary with output realizations.
\end{lemma}

When countries share prior beliefs, $\tau=\tau^*$, efficient-portfolio theory holds. Both countries hold a fraction of the world portfolio, which implies that there is no home bias in portfolio holdings, nor cyclicality in capital flows. This is true even though agents form beliefs using their lifetime experiences. Within a country, however, individual cohorts need not hold the world portfolio. Instead, deviations of young agents from the world portfolio are off-set by deviations of old agents in that same country, so that the efficient-portfolio result holds in aggregate.

The results in Lemma \ref{l: symmetry} have important implications for international flows. They suggest that with symmetric prior beliefs and equal demographics of market participants, domestic, foreign, or local booms or recessions should have no impact on a country's aggregate portfolio holdings, and thus international flows. In what follows, we show how asymmetries between countries overturn this result.


\subsection{Equilibrium Demands Under EBL}

We now allow for countries to not be perfectly symmetric. In particular, we suppose that agents have more precise prior distributions about their own country's output than the foreign country's output, $\tau>\tau^*$. We begin by analyzing the demands of generation $n$ in country $i\in\{H,F\}$ at time $t$, given prices $p_{j,t}$, by solving problem \eqref{eq: Generation n Problem}.

\begin{proposition}\label{p: demands} The demand of generation $n\in\{t,t-1\}$ in country $i\in\{H,F\}$ for the risky asset of country $j\in\{H,F\}$ at time $t$ is
\begin{equation}  \label{eq: risky_demands}
x_{j,t}^{i,n}	=\frac{E_{t}^{i,n}\left[y_{j,t+1}+p_{j,t+1}\right]-R p_{j,t}}{\gamma V_{t}^{i,n}\left[y_{j,t+1}+p_{j,t+1}\right]},
\end{equation}
and the demand for the risk-free asset of cohort $n$ in country $i$ at time $t$ is
\begin{equation}
a^{i,n}_{t}= \mathcal{W}_{t}^{i,n} - x_{H,t}^{i,n} p_{H,t} - x_{F,t}^{i,n} p_{F,t}.
\end{equation}
\end{proposition}

With these demands, we can impose market clearing to solve for market prices of the risky asset of country $j\in\{H,F\}$:
\begin{equation} \label{eq: market_clearing}
\sum_{n\in \{t,t-1\}} \left(  \frac{1}{4} \frac{E_{t}^{H,n}\left[y_{j,t+1} + p_{j,t+1}\right]-R p_{j,t}}{\gamma V_{t}^{H,n}\left[y_{j,t+1}+p_{j,t+1}\right]} + \frac{1}{4} \frac{E_{t}^{F,n}\left[y_{j,t+1}+p_{j,t+1}\right]-R p_{j,t}}{\gamma V_{t}^{F,n}\left[y_{j,t+1}+p_{j,t+1}\right]} \right) = 1
\end{equation}

To solve this, we guess that the price of the asset of country $j$ at time $t$ is given by equation \eqref{e: LinearPrices}. We use the method of undetermined coefficients to verify our guess and solve for the price loadings to fully characterize the behavior of equilibrium asset prices in this economy. Results are summarized in the following Proposition, where we use again
the simplified subindices 0 and 1 for ``the young'' and ``the old'' with subindices 0 and 1 (more generally, an agent's \textit{age}).

\begin{proposition}\label{p: eq_prices}
The equilibrium price of the asset of country $j\in\{H,F\}$ at time $t$ is given by
\begin{equation}
p_{j,t} = \alpha_j + \beta_{j,0} y_{j,t} + \beta_{j,1} y_{j,t-1}
\end{equation}
where {\small
	\begin{align*}
	\beta_{j,0}&=\frac{\Sigma R}{\Sigma R-\frac{1}{4}\left\{ \frac{w_{0}}{\sigma_{0}^{2}}+\frac{w_{1}\omega}{\sigma_{1}^{2}}+\frac{w_{0}^{*}}{\sigma_{0}^{*2}}+\frac{w_{1}^{*}\omega}{\sigma_{1}^{*2}}+\frac{1}{R}\left(\frac{w_{1}\left(1-\omega\right)}{\sigma_{1}^{2}}+\frac{w_{1}^{*}\left(1-\omega\right)}{\sigma_{1}^{*2}}\right)\right\} } -1 \nonumber \\
	\beta_{j,1}&=\left(1+\beta_{i,0}\right)\frac{\left(1-\omega\right)}{4\Sigma R}\left\{ \frac{w_{1}}{\sigma_{1}^{2}}+\frac{w_{1}^{*}}{\sigma_{1}^{*2}}\right\} \\
	\alpha_{j}&=\frac{\left(1+\beta_{i,0}\right)}{\Sigma\left(R-1\right)}\left(\frac{1-w_{0}}{\sigma_{0}^{2}}+\frac{1-w_{1}}{\sigma_{1}^{2}}+\frac{1-w_{0}^{*}}{\sigma_{0}^{*2}}+\frac{1-w_{1}^{*}}{\sigma_{1}^{*2}}\right)m-\frac{\gamma\left(1+\beta_{i,0}\right)^2}{\Sigma\left(R-1\right)} \nonumber
	\end{align*}}
\noindent for $\Sigma \equiv \frac{1}{4}\left(\frac{1}{\sigma_{0}^{2}}+\frac{1}{\sigma_{1}^{2}}+\frac{1}{\sigma_{0}^{*2}}+\frac{1}{\sigma_{1}^{*2}}\right).$
\end{proposition}

Proposition \ref{p: eq_prices} makes two important points. First, the price of output in country $i$ only varies with realizations of output of that same country. This result relies on the assumption that agents (correctly) believe that outputs are independently distributed across countries.
Second, the price of output depends on past output realizations that have been observed by at least one market participant. That is, prices are sensitive to past output realizations that have been experienced by (some) market participants.

With these results in hand, we can compute the holdings of both risky assets for each cohort in each country as a function of past output realizations and analyze the effect of output shocks on cross-country flows.

At time $t$, the holdings of cohort $n\in\{t,t-1\}$ in country $i\in\{H,F\}$ of assets from country $i$ and $j\in\{H,F\}\neq i$ are
{\footnotesize
	\begin{align*}
	x_{i,t}^{i,t}	&=\frac{\tilde{\alpha_{i}}+\left(1+\beta_{0,i}\right)\left(1-w_{0}\right)m}{\gamma\left(1+\beta_{0,i}\right)^{2}\sigma_{0}^{2}}+\frac{w_{0}\left(1+\beta_{0,i}\right)+\beta_{1,i}-R\beta_{0,i}}{\gamma\left(1+\beta_{0,i}\right)^{2}\sigma_{0}^{2}}y_{i,t}-\frac{R\beta_{1,i}}{\gamma\left(1+\beta_{0,i}\right)^{2}\sigma_{0}^{2}}y_{i,t-1} \\
	x_{i,t}^{i,t-1}	&=\frac{\tilde{\alpha_{i}}+\left(1+\beta_{0,i}\right)\left(1-w_{1}\right)m}{\gamma\left(1+\beta_{0,i}\right)^{2}\sigma_{1}^{2}}+\frac{w_{1}\left(1+\beta_{0,i}\right)\omega+\beta_{1,i}-R\beta_{0,i}}{\gamma\left(1+\beta_{0,i}\right)^{2}\sigma_{1}^{2}}y_{i,t}+\frac{w_{1}\left(1+\beta_{0,i}\right)\left(1-\omega\right)-R\beta_{1,i}}{\gamma\left(1+\beta_{0,i}\right)^{2}\sigma_{1}^{2}}y_{i,t-1}  \\
	x_{j,t}^{i,t}	&=\frac{\tilde{\alpha_{j}}+\left(1+\beta_{0,j}\right)\left(1-w_{0}^{*}\right)m}{\gamma\left(1+\beta_{0,j}\right)^{2}\sigma_{0}^{*2}}+\frac{w_{0}^{*}\left(1+\beta_{0,j}\right)+\beta_{1,j}-R\beta_{0,j}}{\gamma\left(1+\beta_{0,j}\right)^{2}\sigma_{0}^{*2}}y_{j,t}-\frac{R\beta_{1,j}}{\gamma\left(1+\beta_{0,j}\right)^{2}\sigma_{0}^{*2}}y_{j,t-1} \\
	x_{j,t}^{i,t-1}	&=\frac{\tilde{\alpha_{j}}+\left(1+\beta_{0,j}\right)\left(1-w_{1}^{*}\right)m}{\gamma\left(1+\beta_{0,j}\right)^{2}\sigma_{1}^{*2}}+\frac{w_{1}^{*}\left(1+\beta_{0,j}\right)\omega+\beta_{1,j}-R\beta_{0,j}}{\gamma\left(1+\beta_{0,j}\right)^{2}\sigma_{1}^{*2}}y_{j,t}+\frac{w_{1}^{*}\left(1+\beta_{0,j}\right)\left(1-\omega\right)-R\beta_{1,j}}{\gamma\left(1+\beta_{0,j}\right)^{2}\sigma_{1}^{*2}}y_{j,t-1}
	\end{align*}
}
where $\tilde{\alpha_{i}}=\alpha_i (1-R)$ and $\omega\equiv\omega(1)=\frac{1}{2}$ as defined in \eqref{eq: posterior_beliefs}. (We have dropped $\omega(0)=1$.)

We now introduce the notion of home bias in portfolio holdings for our two-country economy.

\begin{definition}[Home Bias] We say that there is home bias in portfolio holdings at time $t$ when $$\textit{HB}_{t} \equiv X^H_{H,t}-X^H_{F,t} > 0.$$
\end{definition}

Due to the symmetry embedded in the model, it is sufficient to compare a country's holdings of domestic versus foreign assets to assess the presence of home bias in portfolio holdings. As the world portfolio is given by one unit of the domestic asset and one unit of the foreign asset, no home bias in portfolio holdings would imply $X^H_{H,t} = X^H_{F,t} = X^F_{H,t} = X^F_{F,t} =1$. That is, all countries hold the same fraction of assets of the $H$ and $F$ country, and thus $\textit{HB}_t=0$. Finally, due to market clearing, $\textit{HB}_t>0$ implies that  $X^F_{F,t}-X^F_{H,t}>0$ as well, so it is WLOG to focus on the holdings of country $H$.

We also introduce the notion of booms and recessions in individual countries, defined as periods when positive or negative macro-shocks follow a period of output stability in both countries.

\begin{definition}[Booms and Recessions] We say that there is a recession in country $i$ at time $t$ when $y_{i,t-1}=y_{j,t}=y_{j,t-1}=\bar{y}$ and $y_{i,t}<\bar{y}$. Analogously, we say that there is a boom in country $i$ at time $t$ when $y_{i,t-1}=y_{j,t}=y_{j,t-1}=\bar{y}$ and $y_{i,t}>\bar{y}$.
\end{definition}

We can now show that home bias in portfolio holdings arises in a world with experience-based learners, who believe to have more precise priors about their home country's output.

\begin{proposition}[Home Bias] \label{p: home_bias_priors} If prior beliefs are centered around the truth, i.\,e., $m^i_i=m^j_i=\theta_i$, there will be home bias in portfolio holdings in expectations: $$E\left[X^H_{H,t}-X^H_{F,t}\right] > 0,$$ where $E[\cdot]$ is the unconditional expectations operator over the true distribution of outputs.
\end{proposition}

Intuitively, home bias arises because
a more precise domestic prior implies that domestic agents perceive less risk when investing in the domestic than in the foreign asset: $\sigma^2_{age}<\sigma_{age}^{*2}$, where $age \in \{0,1\}$. As a result, all else equal, they are more willing than foreigners to hold the domestic relative to the foreign asset, generating home bias in portfolio holdings.

A more precise domestic prior belief has a second implication: Domestic agents put more weight on their prior belief about domestic output than foreign agents do when computing their posterior mean, i.\,e., $1-w_{age} > 1-w_{age}^*$. From inspection of posterior means \eqref{eq: posterior_beliefs}, it follows that agents under-react to domestic shocks relative to foreign shocks, $w_{age}<w_{age}^*$, when forming their beliefs. As a result, after a negative shock, we observe an inflow of domestic funds together with an outflow of foreign funds, increasing home bias in portfolio holdings, generating cyclicality of capital flows:

\begin{proposition}[Cyclicality of Flows] \label{p: priors_recessions}
After a recession in country $i$ at time $t$, we observe an inflow of domestic funds, $X^i_{i,t}-X^i_{i,t-1}>0$ (retrenchment),  and an outflow of foreign funds, $X^j_{i,t}-X^j_{i,t-1}<0$ (fickleness). The opposite flow pattern follows a boom in country $i$.
\end{proposition}

Proposition \ref{p: priors_recessions} suggests that our model can rationalize two patterns of international flows that have been identified in the literature \citep{broner2013gross, forbesetal2012} and that are hard to explain with standard models that \emph{also} rationalize home bias. First, there is an inflow of domestic funds during domestic crises --  \emph{retrenchment}. Second, there is an outflow of foreign funds during domestic crises -- \emph{fickleness}.

\cite{broner2006peril} and \cite{milesi2011great} find that during global downturns all countries experience a retrenchment of funds. This is also true in our model. If we define global downturns as periods in which both countries are in a recession, then the model predicts that agents disinvest in foreign assets while investing in domestic assets. To the extent that outputs are correlated in the real world, then the same pattern is likely to be observed after domestic crisis as well. 

\section{Heterogeneity in Market Participation}\label{sec:Heterogeneity}

To analyze the effect of demographics on capital flows, we now allow for different demographics of market participants across countries.
There are many reasons for such differences in participants' demographics. First, the difference could be driven by differences in the country's overall demographics due to cross-country variation in fertility and mortality rates. Second, cultural reasons or market frictions might induce different generations to participate or exit the market at different points in their lives. Our simple model does not formalize the exogenous and endogenous reasons for the differences in market demographics; but we believe that our reduced-form approach captures some of the forces at play when demographics do differ across countries.

In our simple OLG setting, we
denote by $\phi^{i}_{age}$ the mass of agents of $age\in\{0,1\}$ in country $i$ that participate in the market for risky claims. Agents that do not participate in the market invest all of their wealth in the safe asset/technology.\footnote{Formally, in each country, each cohort (young and old) faces a positive probability of being able to access the market for risky asset, so that in equilibrium the demographics of market participants is given by $\{\phi^{i}_{age}\}.$}

Our equilibrium definition then needs to be adjusted, as the new market clearing condition states that in all $t \in \mathbb{N}_{0}$: \begin{align}
1 =& \phi_0^{H}  x^{H,t}_{H,t} + \phi_1^{H} x^{H,t-1}_{H,t} + \phi_0^{F} x^{F,t}_{H,t} + \phi_1^{F} x^{F,t-1}_{H,t} , \\
1 =&  \phi_0^{H}x^{H,t}_{F,t} + \phi_1^{H} x^{H,t-1}_{F,t} + \phi_0^{F} x^{F,t}_{F,t} + \phi_1^{F} x^{F,t-1}_{F,t}.
\end{align}

By using our linear guess for prices, plugging in the posterior means and variances into equation \eqref{eq: market_clearing}, and using the method of undetermined coefficients, we obtain the following adjusted expressions that fully characterize prices for the asset of country $i\in\{H,F\}$ when market participations differ across countries:
{\small
	\begin{align}
	\beta_{i,0}&=\frac{\Sigma R}{\Sigma R-\left\{ \frac{\phi^{i}_0}{\sigma_{0}^{2}}w_{0}+\frac{\phi^{i}_1}{\sigma_{1}^{2}}w_{1}\omega+\frac{\phi^{j}_0}{\sigma_{0}^{*2}}w_{0}^{*}+\frac{\phi^{j}_1}{\sigma_{1}^{*2}}w_{1}^{*}\omega+\frac{\left(1-\omega\right)}{R}\left(\frac{\phi^{i}_1}{\sigma_{1}^{2}}w_{1}+\frac{\phi^{j}_1}{\sigma_{1}^{*2}}w_{1}^{*}\right)\right\} } - 1 \nonumber \\
	\beta_{i,1}&=\frac{\left(1+\beta_{i,0}\right)\left(1-\omega\right)}{\Sigma R}\left\{ \frac{\phi^{i}_1}{\sigma_{1}^{2}}w_{1}+\frac{\phi^{j}_1}{\sigma_{1}^{*2}}w_{1}^{*}\right\} \\
	\alpha_{i}&=\frac{\left(1+\beta_{i,0}\right)}{\Sigma\left(R-1\right)}\left(\frac{\phi^{i}_0}{\sigma_{0}^{2}}\left(1-w_{0}\right)+\frac{\phi^{i}_1}{\sigma_{1}^{2}}\left(1-w_{1}\right)+\frac{\phi^{j}_0}{\sigma_{0}^{*2}}\left(1-w_{0}^{*}\right)+\frac{\phi^{j}_1}{\sigma_{1}^{*2}}\left(1-w_{1}^{*}\right)\right)m-\frac{\gamma\left(1+\beta_{i,0}\right)^2}{\Sigma\left(R-1\right)} \nonumber
	\end{align}}
where $\Sigma \equiv \frac{\phi^{i}_0}{\sigma_{0}^{2}}+\frac{\phi^{i}_1}{\sigma_{1}^{2}}+\frac{\phi^{j}_0}{\sigma_{0}^{*2}}+\frac{\phi^{j}_1}{\sigma_{1}^{*2}}$ and all other price loadings are zero.\footnote{These results are derived in the Proof of Proposition \ref{p: eq_prices} in Appendix \ref{app:Proofs}}

Portfolio holdings and prices depend on the mass of young and old agents in countries $H$ and $F$, as different agents have different prior beliefs and experiences, and thus their weights matter for aggregate asset demand. We compute the portfolio holdings of each cohort in equilibrium and analyze how different demographics of market participation affect the sensitivity of flows to output shocks that was characterized in Proposition \ref{p: priors_recessions}.

\begin{proposition} \label{p: demographics} There exist thresholds $\bar{\tau_1},\bar{\tau_2}\in(\tau^*,\sigma^{-2})$ such that the cyclicality of capital flows, i.\,e., the sensitivity of domestic portfolio holdings to domestic output shocks, varies with demographics as follows
	{\small
\begin{align}
\frac{\partial}{\partial\phi_{0}^{j}}\left(\frac{\partial X_{i,t}^{i}}{\partial y_{i,t}}\right)&=\frac{\phi_{0}^{i}}{\sigma_{0}^{2}\sigma_{0}^{*2}}\left(w_{0}-w_{0}^{*}\right)+\frac{\phi_{1}^{i}}{\sigma_{0}^{*2}\sigma_{1}^{2}}\left(w_{1}\omega-w_{0}^{*}\right)<0. \label{e: derdem_1}\\
\frac{\partial}{\partial\phi_{1}^{j}}\left(\frac{\partial X_{i,t}^{i}}{\partial y_{i,t}}\right)&=\frac{\phi_{0}^{i}}{\sigma_{0}^{2}\sigma_{1}^{*2}}\left(w_{0}-w_{1}^{*}\omega\right)+\frac{\phi_{1}^{i}}{\sigma_{1}^{2}\sigma_{1}^{*2}}\left(w_{1}-w_{1}^{*}\right)\omega>0, \text{if and only if } \tau<\bar{\tau_1}. \label{e: derdem_2}
\\\frac{\partial}{\partial\phi_{0}^{i}}\left(\frac{\partial X_{i,t}^{i}}{\partial y_{i,t}}\right)&=\frac{\phi_{0}^{j}}{\sigma_{0}^{2}\sigma_{0}^{*2}}\left(w_{0}-w_{0}^{*}\right)+\frac{\phi_{1}^{j}}{\sigma_{0}^{2}\sigma_{1}^{*2}}\left(w_{0}-w_{1}^{*}\omega\right) >0,\text{if and only if } \tau<\bar{\tau_2}.  \label{e: derdem_3} \\
\frac{\partial}{\partial\phi_{1}^{i}}\left(\frac{\partial X_{i,t}^{i}}{\partial y_{i,t}}\right)&=\frac{\phi_{0}^{j}}{\sigma_{0}^{*2}\sigma_{1}^{2}}\left(w_{1}\omega-w_{0}^{*}\right)+\frac{\phi_{1}^{j}}{\sigma_{1}^{2}\sigma_{1}^{*2}}\left(w_{1}-w_{1}^{*}\right)\omega<0. \label{e: derdem_4}
\end{align}}
with the opposite sign for changes in the sensitivity of foreign portfolio holdings to domestic output shocks, i.\,e., $\frac{\partial X_{j,t}^{i}}{\partial y_{i,t}}.$

\end{proposition}

The first important result from Proposition \ref{p: demographics} is that portfolios holdings of the domestic asset, and thus capital flows, are more sensitive to domestic output shocks as the mass of young foreign agents increases, equation \eqref{e: derdem_1}, or equivalently as the mass of old domestic agents increases, equation \eqref{e: derdem_4}. This is because young foreign agents react the most to such shocks, increasing prices and thus causing a larger decrease in domestic demand for the domestic asset. Old domestic agents instead react the least, further reducing the response of domestic demand for the domestic asset. As for the effect of the young domestic and old foreign agents, the result is ambiguous. This is because it now matters who has a more precise posterior: young domestic agent with no experience but precise priors, or old foreign agents with experience but less precise priors. The conditions on the domestic precision indicate exactly the point in which the effect of a more precise prior precision dominates, which implies that a higher number of domestic young agents or of old foreign agents dampens the sensitivity of capital flows to domestic output shocks.

\section{Empirical Implications} \label{sec:EmpiricalImplications}

The theoretical framework illustrates that the concept of experience-based learning, which has been shown to play a significant role in investor belief formation in the domestic context, may also help explain international capital flows. In this section, we turn to the empirical regularities of home bias and the cyclicality of capital flows, and test whether the predictions of EBL regarding the relation of these regularities with demographics hold in the data.

First, we re-estimate the existence and extent of home bias, in line with results in Proposition \ref{p: home_bias_priors}. Second, we analyze the cyclicality of capital flows, as predicted by Proposition \ref{p: priors_recessions}: Is there an inflow of domestic funds (retrenchment) and an outflow of foreign funds (fickleness) during domestic crises?
Finally, we test if the sensitivity of capital flows to output shocks varies with demographics as predicted by Proposition \ref{p: demographics}: Are portfolio holdings of domestic assets, and thus capital flows, more sensitive to domestic output shocks when the fraction of the older (domestic) population is higher?

\subsection{Data}

Our main sources of data are the International Monetary Fund's (IMF) International Financial Statistics (IFS), including its International Investment Positions (IIP) data; the IMF's Balance of Payments (BOP) data, the World Bank's World Development Indicators (WDI), and the United Nations' (UN) data on World Population Prospects (WPP). 

The main data series extracted from these data sets for the empirical analysis are yearly measures of Home Bias in the US, capital outflows by domestic agents (COD), capital inflows by foreign agents (CIF), and the growth of the Gross Domestic Product (GDP growth). We define and discuss all measures below. More details about the sources of data and the construction of the key variables are in Appendix \ref{app:Data}, including some summary statistics in Table \ref{tab:summary}.

\subsection{Home Bias}\label{ss:HomeBias}

Equity home bias is ``the empirical phenomenon that investors’ portfolios are concentrated in domestic equities to a much greater degree than justified by portfolio theory'' \citep{cooper2013equity}. It has been well documented in the literature both at an aggregate level \citep{cooper1986costs, french1991investor, cooper2013equity, mishra2015measures} and at the fund level \citep{hau2008home}. Experience-based learning offers an explanation for the persistence of home bias in equity as well as the variation in the extent of home bias over time: 
It should increase when a country has a negative growth/GDP shock. Below we will test the (refined) hypotheses that domestic investors increase their investment in domestic assets (less domestic outflows) after a negative GDP growth shock and that foreigners decrease their investment in domestic assets (less foreign inflows) after a negative GDP growth shock, which will imply that home bias increases. 

For completeness, we start by providing evidence on the existence of equity home bias throughout our sample period, which extends prior evidence on equity home bias. These estimations aim to confirm its persistence throughout a longer time period and for more recent years than established in prior literature.

We focus on US home bias and map our two-country model to the US versus the rest-of-the world (ROW). We measure US home bias as the difference between the percentage of US equity held domestically (i.\,e., domestic equity investment by the US divided by the global equity investment by the US) and a measure of optimal asset location. For the latter measure, we use the US share of global market capitalization, which is the traditional benchmark in prior research on home bias \citep{cooper1986costs,brealey1999international,berkel2007institutional,cooper2013equity}.
Thus, home bias is calculated as:
$$\frac{\textrm{Domestic Equity Investment in the US}}{\textrm{Global Equity Investment by the US}} - \frac{\textrm{US Market Capitalization}}{\textrm{Global Market Capitalization}}.$$ 

We obtain data on domestic and foreign equity investment in the US and globally from the IIP data of the IMF, available for 1976-2017, and capitalization data from the World Federation of Exchanges, available from the World Bank's WDI data for 1980-2017.

Table \ref{tab:summary} provides the summary statistics of the resulting measure of home bias.
We see that home bias has been large and persistent with a minimum value of 32\%, a mean of 44\%, and a standard deviation of only 9\%.

\begin{figure*}[h]
		\vspace{0.4cm}
	\centering     
	\includegraphics[width=0.6\textwidth]{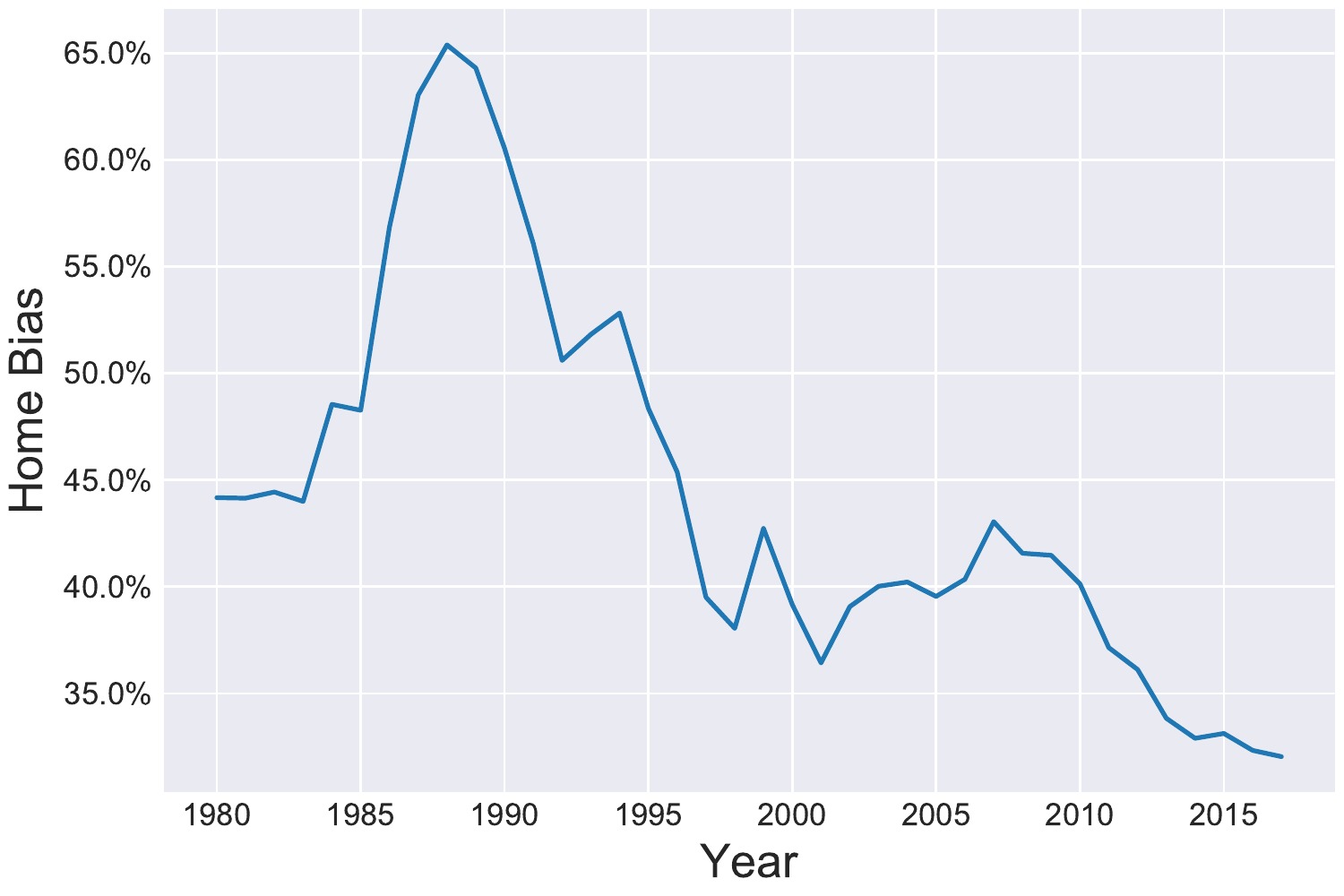}\hfill
	\caption{US Home Bias by Year}
	\label{f:home_bias}
\end{figure*}

To illustrate the extent of US home bias over our sample period, we plot US home bias from 1980 to 2017 in Figure \ref{f:home_bias}. As the graph shows, home bias has remained large and persistent over our 38-year sample period, never dropping below 30\% in our almost 40-year sample. For the analysis in this paper, we do not focus on the year-to-year changes. The influence of experience-based learning is measurable over medium- to long-term horizons (as discussed in \cite{malmendierpouzovanasco2018}), and we abstract from short-term factors such as market revaluation in the regression analyses shown in the next subsection.

\subsection{Cyclicality of Capital Flows: Fickleness and Retrenchment}

Fickleness is defined as the tendency of foreign investors to exit when a country is in distress, and retrenchment as the tendency of local investors to reduce their foreign investments and ``bring home their global liquidity'' during domestic distress \citep{caballero2016model}.
These patterns have been documented in prior literature, although they are less well established than home bias. Proposition \ref{p: priors_recessions} shows that experience-based learning offers an explanation for fickleness and entrenchment. 
Here, we corroborate the empirical patterns using more current data and a different panel structure.

We construct measures of capital flows from yearly changes in a country's investments abroad, assets, and foreign countries' investments domestically, liabilities, in the Balance of Payments (BOP) data of the IMF.\footnote{\cite{broner2013gross} used inflows and outflows variables, which are included in the 5th edition of the BOP manual. We follow \cite{caballero2016model} and use the 6th edition, which defines variables in terms of assets and liabilities instead of outflows and inflows.} 
To capture changes in foreign countries' investments domestically, we calculate capital inflows by foreign agents (CIF) 
as the sum of changes in a country's liabilities from direct investment, portfolio investment, and other investment from the financial account. To capture changes in the domestic country's investments abroad, we calculate capital outflows of domestic agents (COD) as the sum of changes in a country's assets from direct investment, portfolio investment, other investment, and reserve assets from the financial account. We include the 19 countries represented in the G20, 
as these countries are all of a certain size, involved with international finance, and notable enough so that news of income shocks will reach the global investment community. 
We inflation-adjust all capital-flow values and, following \cite{broner2013gross} and \cite{caballero2016model}, normalize them by trend GDP and standardize to a mean of zero and a standard deviation of one (at the level of each country). The resulting data series runs from 1970-2017, varying by country.

To proxy for domestic crises, or output shocks, we follow \cite{broner2013gross} and focus on how output varies during the business cycle (rather than calculating indicators for income shocks across countries, which are hard to reliably define).
We use yearly GDP data in contemporaneous US dollars from the World Bank, running from 1970-2017 (varying by country). 
To calculate real GDP growth, we first inflation adjust the GDP 
data,\footnote{As Balance of Payments (BOP) data is recorded in contemporaneous US dollars, as opposed to the local currency unit, we start with GDP data in contemporaneous US dollars so that our calculations are parallel.} and then calculate current-year real GDP minus prior-year real GDP, all divided by prior-year real GDP. The resulting panel of real GDP growth runs from 1971-2017 (varying by country), 
since we lack a prior year for the 1970 calculation.
We note that, while quarterly GDP data is also available (from the IFS), we use yearly data from the WDI, mirroring \cite{broner2013gross}, as the annual data provides better coverage across time for most countries.

For further details on data sources, see Appendix \ref{app:Data}.

\smallskip

With these variables at hand, we test Proposition \ref{p: priors_recessions} in our new panel.
Proposition \ref{p: priors_recessions} predicts that, during a bust, domestic investors retrench by pulling money out of foreign countries to invest at home, and foreign investors become ``fickle'' and pull money out of the domestic country. The opposite effects hold during booms.
These predictions imply that foreign investment in the domestic country and domestic investment abroad are both pro-cyclical, i.\,e., positively correlated with GDP growth or other measures of the business cycle.

To test these predictions, and confirm the results of \cite{broner2013gross} for our extended time period and the selected set of countries, we estimate the following empirical model:
$$ Y_{i,t} = \alpha_i + \gamma_i \cdot t + \beta_1 \cdot \% \Delta \textrm{GDP}_{i,t} + \epsilon_{i,t} ,$$
where $i$ indexes countries, $t$ indexes/represents the time period, $Y_{i,t}$ is the flow variable for a given country year, $\alpha_i$ is the country fixed effect, $\gamma_i$ is the linear time trend coefficient, $\beta_1$ is the coefficient of interest, $\% \Delta \textrm{GDP}_{i,t}$ is real GDP growth between $t$ and $t-1$, and $\epsilon_{i,t}$ is the error term. 

Our model of experience based learning predicts a positive coefficient estimate $\beta_1$ on GDP growth, both when we use the CIF measure of capital inflows of foreigners and when we use the COD measure of capital outflows of domestics as the outcome variable $Y_{i,t}$. Moreover, this prediction is consistent with the results of \cite{broner2013gross} (in their Table 3) that both CIF and COD correlate positively with GDP, implying that both CIF and COD are pro-cyclical. 
%

\begin{table}[htbp]
	\caption{Business Cycles and Flows}
	\vspace{-0.4cm}
	\begin{small}
		\begin{tablenotes}
			\item CIF are the sum of changes in a country's liabilities from direct investment, portfolio investment, and other investment from the financial account divided by trend GDP and standardized at the country level. COD are the sum of changes in a country's assets from direct investment, portfolio investment, other investment, and reserve assets from the financial account divided by trend GDP and standardized at the country level. Real GDP Growth is the percent GDP growth.
		\end{tablenotes}
	\end{small}
	\vspace{0.5em}
	\begin{center}
		{
		\def\sym#1{\ifmmode^{#1}\else\(^{#1}\)\fi}
		\begin{tabular}{l*{2}{c}}
		\hline\hline
		                    &\multicolumn{1}{c}{(1)}&\multicolumn{1}{c}{(2)} \Tstrut\\
		                    &\multicolumn{1}{c}{CIF}&\multicolumn{1}{c}{COD}\\
		\hline
		Real GDP Growth     &       0.832\sym{**} &       1.852\sym{***} \Tstrut\\
		                    &      (2.68)         &      (7.64)         \Bstrut\\
		\hline
		Country FE          &           X         &           X         \Tstrut\\
		Country Linear Trend&           X         &           X         \\
		Country Clustered SE&           X         &           X         \\
		Income Group        &         G20         &         G20         \Bstrut\\\hline
		R2                  &       0.188         &       0.261         \Tstrut\\
		N                   &         710         &         720         \\
		Number of Countries      &          19         &          19         \\
		\hline\hline
		\multicolumn{3}{l}{\footnotesize \textit{t} statistics in parentheses}\\
		\multicolumn{3}{l}{\footnotesize \sym{*} \(p<0.10\), \sym{**} \(p<0.05\), \sym{***} \(p<0.01\)}\\
		\end{tabular}
		}

	\end{center}
	\label{tab:broner_rep}
\end{table}

Table \ref{tab:broner_rep} shows the regression results, with the CIF measure of capital flows as the outcome variable in column (1), and the COD measure in column (2). 
As the estimates reveal, we find the predicted positive relationship with our measure for the
business cycle booms, real GDP  growth.
In both regressions, the estimates for the coefficient of interest are significant, controlling for country fixed effects and a country specific linear trend.

\subsection{Cyclicality of Capital Flows and Demographics}

Having confirmed the persistence of home bias, fickleness, and retrenchment in our data, we now move to the implications of Proposition \ref{p: demographics} which are novel and specific to the EBL framework. 

An empirical prediction arising from Proposition \ref{p: demographics} 
is that an increase in the older domestic population of market participants increases the business-cycle sensitivity of flows.
This is because the domestic old population has the strongest prior among all demographic groups. As a result, a domestic recession has the least impact on the beliefs of older domestic investors.
In other words, older domestic investors believe that other demographic groups are underestimating domestic production and that the domestic asset is under-priced, leading to retrenchment.\footnote{Note that, for these comparative statics, we are not assuming that the total population remains constant: An increase in the number of old agents is not equivalent to a decrease in the number of young agents.}

Proposition \ref{p: demographics} does not offer a prediction whether a larger young domestic population will increase or decrease the flow sensitivity to business cycle fluctuations. The reason is that it is unclear how the domestic young agents react relative to the foreign old agents: The former are endowed with a more precise prior, but the latter group has more experience.
We will explore the relationship as an empirical question in the data.

In order to perform these empirical tests, 
we introduce four country-level indicator variables for both median and top-quartile population statistics, and the old and young population.  These indicator variables equal one if a country has an above median, or (alternatively) a top-quartile level of young or old investors, respectively.
Following \cite{malmendierpouzovanasco2018}, who consider those aged 25 to 75 the investor population, we define the young investor population of a country as those aged 25 to 49 and the old investor population as those aged 50 to 74.
To measure population mass, we use data on yearly, country-level population data in 5-year age bins from the World Population Prospects (WPP) provided by the United Nations (UN). More details are in Appendix \ref{app:Data}.

To test our prediction, we extend the regression specification of Table \ref{tab:broner_rep}. As before, we regress a flow variable, CIF or COD, on real GDP growth, a country fixed effect, and a country linear time trend, but we add interactions of GDP growth with indicators for high old investor population and, for completenesss, high young investor populations. The regression specification is
$$ Y_{i,t} = \alpha_i + \gamma_i \cdot t + \beta_1 \cdot \% \Delta \textrm{GDP}_{i,t} + \beta_2 \cdot \% \Delta \textrm{GDP}_{i,t} \cdot D_{i,t} + \epsilon_{i,t} ,$$
where the relevant parameters follow the above definitions, and $\beta_2$ is the new coefficient of interest, with $D_{i,t}$ being one of the indicator variables described above.

In Table \ref{tab:broner_gdp}, columns (1) and (2), we add two interactions of real GDP growth: with indicators for old investor populations above the median and young investor populations above the median. In columns (3) and (4), we add two interactions of real GDP growth with indicators for old investor populations in the top quartile and young investor populations in the top quartile.

\begin{table}[htbp]
	\caption{Population and Business Cycle Sensitivity}
	\vspace{-0.4cm}
	\begin{small}
		\begin{tablenotes}
			\item CIF are the sum of changes in a country's liabilities from direct investment, portfolio investment, and other investment from the financial account divided by trend GDP and standardized at the country level. COD are the sum of changes in a country's assets from direct investment, portfolio investment, other investment, and reserve assets from the financial account divided by trend GDP and standardized at the country level. Real GDP Growth is the percent GDP growth.
		\end{tablenotes}
	\end{small}
	\vspace{0.5em}
	\begin{center}
		{
			\def\sym#1{\ifmmode^{#1}\else\(^{#1}\)\fi}
			\begin{tabular}{l*{4}{c}}
				\hline\hline
				&\multicolumn{1}{c}{(1)}&\multicolumn{1}{c}{(2)}&\multicolumn{1}{c}{(3)}&\multicolumn{1}{c}{(4)} \Tstrut\\
				&\multicolumn{1}{c}{CIF}&\multicolumn{1}{c}{COD}&\multicolumn{1}{c}{CIF}&\multicolumn{1}{c}{COD}\\
				\hline
				Real GDP Growth                       &       0.963\sym{*}  &       1.834\sym{***}&       0.702\sym{*}  &       1.549\sym{***}\Tstrut\\
				&      (2.08)         &      (8.20)         &      (1.94)         &      (5.19)         \\
				[1em]
				Old Pop. Above Median X GDP growth     &       0.471         &       1.916\sym{***}&                     &                     \\
				&      (0.47)         &      (5.07)         &                     &                     \\
				[1em]
				Young Pop. Above Median X GDP growth   &      -0.683         &      -1.692\sym{***}&                     &                     \\
				&     (-0.69)         &     (-5.82)         &                     &                     \\
				[1em]
				Old Pop. in Top Quartile X GDP growth  &                     &                     &       1.241         &       2.046\sym{***}\\
				&                     &                     &      (1.50)         &      (3.92)         \\
				[1em]
				Young Pop. in Top Quartile X GDP growth&                     &                     &      -0.280         &       0.491         \\
				&                     &                     &     (-0.45)         &      (1.16)         \Bstrut\\
				\hline
				Country FE                            &           X         &           X         &           X         &           X         \Tstrut\\
				Country Linear Trend                  &           X         &           X         &           X         &           X         \\
				Country Clustered SE                  &           X         &           X         &           X         &           X         \\
				Country Sample                        &         G20         &         G20         &         G20         &         G20         \Bstrut\\\hline
				R2                                    &       0.188         &       0.267         &       0.191         &       0.271         \Tstrut\\
				N                                     &         710         &         720         &         710         &         720         \\
				Number of Countries                   &          19         &          19         &          19         &          19         \\
				\hline\hline
				\multicolumn{5}{l}{\footnotesize \textit{t} statistics in parentheses}\\
				\multicolumn{5}{l}{\footnotesize \sym{*} \(p<0.10\), \sym{**} \(p<0.05\), \sym{***} \(p<0.01\)}\\
			\end{tabular}
		}
	\end{center}
	\label{tab:broner_gdp}
\end{table}

In all four columns, we estimate a positive coefficient of the interaction between the indicators for above-median (or top-quartile) old investor populations and GDP growth.
The coefficient is significant at the 1\% level in columns (2) and (4) where COD is the dependent variable.
That is, consistent with Proposition \ref{p: demographics}, we find that having a larger domestic population of older investors predicts increased business-cycle sensitivity.
We note that, although COD and CIF are correlated \citep{broner2013gross,caballero2016model,forbesetal2012}, it is not surprising that COD (capital outflows by domestic agents) exhibits a stronger correlation with these interaction variables in light of our theory, as the analysis focuses on the domestic population.

Our estimates also reveal that, empirically, an increase in the domestic young population reduces the sensitivity to business-cycle fluctuations.
In columns (1) to (3) of Table \ref{tab:broner_gdp}, the sign on the interaction with the young population is negative and, in column (2), significant at the 1\% level. These results suggest that the effect of being young, implying a less precise prior, dominates the effect of the stronger domestic prior. As a result, during a domestic bust the domestic young are more likely to invest abroad which would reduce the degree of aggregate retrenchment.

Taken together with the results of our model, these results suggests that, in a given country, the older domestic investors have the most precise prior, the older foreign investors have the second most precise prior, the domestic young have the next most precise prior, and the foreign young have the least precise prior.

\section{Conclusion} \label{sec: Conclusions}

This paper introduces a novel explanation for several classic international macro puzzles surrounding portfolio investment and capital flows. We introduce a modern macro-finance approach to belief formation, experience-based learning (EBL), into a two country CARA-Normal OLG model. In an equilibrium where both countries have the same demographics but prior beliefs about domestic and foreign output differ, portfolio holdings, prices, and demand for a country's risky asset depend on the history of the country's output. Further, this model setting generates home bias in portfolio holdings because domestic agents believe that domestic output is less risky than foreign output. In addition, domestic agents will overreact to foreign output shocks driving two patterns of behavior, fickleness and retrenchment. The effect of a negative foreign output shock will negatively affect the prior of domestics more strongly than their foreign counterparts, 
which induces them to be fickle by selling their foreign risky asset.
Similarly, in the case of a domestic output shock foreign agents will value the domestic asset less than the domestic agents, and this will lead domestic agents to retrench by selling the foreign risky assets to buy the domestic risky asset. Our model also generates additional predictions, in particular that an increase in the domestic old population generates a higher degree of fickleness and retrenchment. We take this prediction to the data and find supportive evidence.

We view this paper as a first step on bringing to bear insights from experience-based learning to shed light on empirical regularities in international macro-finance. It could be extended in several directions. For instance one could consider a subjective model over the process of output and prices, where agents use past realization of all variables to update their beliefs.

\newpage

\singlespacing
\bibliographystyle{junpan}
\bibliography{MacroFinance}


\pagebreak
\begin{appendices}
	\counterwithin{table}{section}
	\setcounter{figure}{0}

	\section{Proofs for Results}\label{app:Proofs}

		\begin{remark}[Market Participation]
		When possible, the proofs are done for a general demographic structure of market participants. Thus, we introduce the following notation: $\phi_{age}^{i}$ is the mass of agents with $age\in\{0,1\}$, i.\,e., young and old, in country $i\in\{H,F\}$. See Section \ref{sec:Heterogeneity} for a discussion on what this captures in our OLG setting.
	\end{remark}

\medskip

\begin{proof}[Proof of Proposition \ref{p: posterior_beliefs}]

	Let $\Sigma_{0} = [1/\tau^{i}_{i},0;0,1/\tau^{i}_{j}]$ and $\theta_{0} = (m^{i}_{i},m^{j}_{j})$. It follows that, for any $A \subseteq \Theta$ Borel,
	\[
	\mu_{n+age}^{i,n}(A)=\frac{\int_{A}\exp\left\{ -0.5\left(\sum_{k=n}^{n+age}\varpi(k)(y_{k}-\theta)^{T}\Sigma^{-1}(y_{k}-\theta)+(\theta-\theta_{0})^{T}\Sigma_{0}^{-1}(\theta-\theta_{0})\right)\right\} d\theta}{\int_{\Theta}\exp\left\{ -0.5\left(\sum_{k=n}^{n+age}\varpi(k)(y_{k}-\theta)^{T}\Sigma^{-1}(y_{k}-\theta)+(\theta-\theta_{0})^{T}\Sigma_{0}^{-1}(\theta-\theta_{0})\right)\right\} d\theta}.
	\]

	Observe that
	\begin{align*}
	& \sum_{k=n}^{n+age}\varpi(k)\left\{ (y_{k})^{T}\Sigma^{-1}(y_{k})-2(y_{k})^{T}\Sigma^{-1}\theta\right\} +\theta^{T}\left(\Sigma^{-1}(age+1)+\Sigma_{0}^{-1}\right)\theta-2\theta_{0}^{T}\Sigma_{0}^{-1}\theta\\
	= & \sum_{k=n}^{n+age}\varpi(k)(y_{k})^{T}\Sigma^{-1}(y_{k})-2\left(\sum_{k=n}^{n+age}\varpi(k)(y_{k})^{T}\Sigma^{-1}+\theta_{0}^{T}\Sigma_{0}^{-1}\right)\theta\\
	& +\theta^{T}\left(\Sigma^{-1}(age+1)+\Sigma_{0}^{-1}\right)\theta.
	\end{align*}

	Letting $\Xi_{age} =  (\Sigma^{-1}(age+1)+\Sigma_{0}^{-1})^{-1}$ and $\alpha^{T}_{n,age} \equiv \left(\sum_{k=n}^{n+age}\varpi(k)(y_{k})^{T}\Sigma^{-1}+\theta_{0}^{T}\Sigma_{0}^{-1}\right)$, it follows that the RHS equals
	\begin{align*}
	& \theta^{T} \Xi^{-1}_{age} \theta - 2\left\{ \alpha^{T}_{n,age} \Xi_{age} \right\} \Xi^{-1}_{age} \theta + \alpha^{T}_{n,age} \Xi_{age}  \alpha_{n,age}  +stuff_{n,age}\\
	= &  \left( \theta - \alpha^{T}_{n,age} \Xi_{age} \right)^{T} \Xi^{-1}_{age}  \left( \theta - \alpha^{T}_{n,age} \Xi_{age} \right) +  stuff_{n,age}
	\end{align*}
	where $stuff_{n,age} \equiv \theta_{0}^{T}\Sigma_{0}^{-1}\theta_{0} + \sum_{k=n}^{n+age}\varpi(k)(y_{k})^{T}\Sigma^{-1}(y_{k}) - \alpha^{T}_{n,age} \Xi_{age}  \alpha_{n,age}$ and does not depend on $\theta$. Therefore,
	\begin{align*}
	\mu_{n+age}^{i,n}(A)=\frac{\int_{A}\exp\left\{ -0.5\left(  \left( \theta - \alpha^{T}_{n,age} \Xi_{age} \right)^{T} \Xi^{-1}_{age}  \left( \theta - \alpha^{T}_{n,age} \Xi_{age} \right)  \right)\right\} d\theta}{\int_{\Theta}\exp\left\{ -0.5\left(  \left( \theta - \alpha^{T}_{n,age} \Xi_{age} \right)^{T} \Xi^{-1}_{age}  \left( \theta - \alpha^{T}_{n,age} \Xi_{age} \right)  \right)\right\} d\theta}.
	\end{align*}
	Which implies that $\mu^{i,n}_{n+age}$ is characterized by a Gaussian pdf with mean $\alpha^{T}_{n,age} \Xi_{age}$ and covariance matrix $\Xi_{age}$.

	Observe that $\Xi_{age}$ is a diagonal matrix with elements given by
	\begin{align*}
	&\sigma^{i,n}_{i,n+age} = \frac{ \sigma^{2}}{(age+1) + \tau^{i}_{i} \sigma^{2}}\\
	&\sigma^{i,n}_{j,n+age} = \frac{ \sigma^{2}}{(age+1) +  \tau^{i}_{j} \sigma^{2}}.
	\end{align*}
	And the mean $\alpha^{T}_{n,age} \Xi_{age}$ is thus given by
	\begin{align*}
	 \left[
	\begin{array}{c}
	\sigma^{-2} \sum_{k=n}^{n+age}\varpi(k)y_{i,k}+ m^{i}_{i} \tau^{i}_{i}\\
	\sigma^{-2} \sum_{k=n}^{n+age}\varpi(k)y_{j,k}+ m^{i}_{j} \tau^{i}_{j}
	\end{array}
	\right]
	\Xi_{age} = \left[
	\begin{array}{c}
	\frac{(age+1)}{(age+1) +  \tau^{i}_{i} \sigma^{2}} \frac{1}{age+1}\sum_{k=n}^{n+age}\varpi(k)y_{i,k}+ m^{i}_{i} \frac{\tau^{i}_{i}\sigma^{2}}{(age+1) + \tau^{i}_{i} \sigma^{2}}\\
	\frac{(age+1)}{(age+1) + \tau^{i}_{j} \sigma^{2}} \frac{1}{age+1}\sum_{k=n}^{n+age}\varpi(k)y_{j,k}+ m^{i}_{j} \frac{\tau^{i}_{j}\sigma^{2}}{(age+1) +\tau^{i}_{j} \sigma^{2}}
	\end{array}
	\right]
	\end{align*}

\end{proof}

\medskip

	\begin{proof}[Proof of Lemma \ref{l: symmetry}] As $\tau=\tau^*$, we have that $w_{age}=w_{age}^{*}$ and that $\sigma_{age}^{2}=\sigma_{age}^{*2}$ for $age\in\{0,1\}$. Using the results from the Proof of Proposition \ref{p: demands}, we have that the aggregate demand for the asset of country $F$ by agents in country $H$ is
		\begin{equation}
		X_{F,t}^{H}=\phi_0^{H} \frac{E^{H,t}_{t}[s_{F,t+1}] }{V^{H,t}_{t}[s_{F,t+1}]} + \phi_1^{H} \frac{E^{H,t-1}_{t}[s_{F,t+1}] }{V^{H,t-1}_{t}[s_{F,t+1}]}
		\end{equation}
	where $s_{j,t+1}\equiv y_{j,t+1}+p_{j,t+1}-Rp_{j,t}$ for $j\in\{H,F\}$. Analogously, the aggregate demand for the asset of country $F$ by agents in country $F$ is
		\begin{equation}
		X_{F,t}^{F}=\phi_0^{F} \frac{E^{F,t}_{t}[s_{F,t+1}] }{V^{F,t}_{t}[s_{F,t+1}]} + \phi_1^{F} \frac{E^{F,t-1}_{t}[s_{F,t+1}] }{V^{F,t-1}_{t}[s_{F,t+1}]}
		\end{equation}
		As both countries have the same prior beliefs about output in country $F$ and all agents in a given cohort observe the same output realizations, it follows that agents in both countries in a given cohort have the same posterior beliefs.

		When both countries have the demographics of market participation, i.\,e., $\phi_{age}^H=\phi_{age}^F$ for all $age\in\{0,1\}$ with $\phi_0^{H}+\phi_1^{H}=\phi_0^{F}+\phi_1^{F}=\phi$, it follows that $X_{F,t}^{H}=X_{F,t}^{F}$. Market clearing implies that $$\phi X_{F,t}^{H} + \phi X_{F,t}^{F}=1.$$ It follows that $X_{F,t}^{H} = X_{F,t}^{F} = \frac{1}{2\phi}.$ As in our baseline model we have assumed that each country has a mass of $\phi=\frac{1}{2}$ agents, the result follows. The proof for  $X_{H,t}^{H} = X_{H,t}^{F} = 1$ is isomorphic.
	\end{proof}

	\medskip

	\begin{proof}[Proof of Proposition \ref{p: demands}]
		Given their posterior belief, the problem that generation $n$ born in country $i$ faces at time $t$ is given by (we drop the superscript $n,i$ to reduce notation):

\begin{align}
\max_{x_{H,t},x_{F,t},b_{t}}&E_{t}[-\exp(-\gamma\cdot \mathcal{W}_{t+1})] \\
s.t.\quad \mathcal{W}_{t+1}&=x_{H,t}\left(y_{H,t+1}+p_{H,t+1}-Rp_{H,t}\right)+x_{F}\left(y_{F,t+1}+p_{F,t+1}-Rp_{F,t}\right) \\
b_{t}&=\mathcal{W}_{t-1}-x_{H,t}p_{H,t}-x_{F,t}p_{F,t}
\end{align}

Due to the linearity of prices (we focus on Affine Equilibria), we know that $\mathcal{W}_{t+1}$ is normally distributed. As a result, the problem can be re-written as follows,
\begin{align}
 \max_{x_{H,t},x_{F,t}}&E_t\left[\mathcal{W}_{t+1}\right]-\frac{1}{2}\gamma V_t\left[\mathcal{W}_{t+1}\right] \\
		s.t.\quad \mathcal{W}_{t+1}&=x_{H,t}\left(y_{H,t+1}+p_{H,t+1}-Rp_{H,t}\right)+x_{F,t}\left(y_{F,t+1}+p_{F,t+1}-Rp_{F,t}\right)
\end{align}
		where

\begin{align*}
E_{t}\left[\mathcal{W}_{t+1}\right]&=x_{H,t}\left(E_{t}\left[y_{H,t+1}+p_{H,t+1}\right]-Rp_{H,t}\right)+x_{F,t}\left(E_{t}\left[y_{F,t+1}+p_{F,t+1}\right]-Rp_{F,t}\right) \\
V_{t}\left[\mathcal{W}_{t+1}\right]&=x_{H,t}^{2}V_{t}\left[y_{H,t+1}+p_{H,t+1}\right]+x_{F,t}^{2}E_{t}\left[y_{F,t+1}+p_{F,t+1}\right]
\end{align*}

as outputs are independently distributed across countries. From the FOC of this problem, we obtain that the demand of generation $n\in\{t,t-1\}$ in country $i\in\{H,F\}$ for the asset in country $j\in\{H,F\}$ at time $t$ is:

$$x_{j,t}^{n,i}	=\frac{E_{t}^{n,i}\left[y_{j,t+1}+p_{j,t+1}\right]-Rp_{j,t}}{\gamma V_{t}^{n,i}\left[y_{j,t+1}+p_{j,t+1}\right]}.$$

And the demand for the risk-free asset of cohort $n$ in country $i$ at time $t$ follows from the budget constraint:

$$b_{t}^{n,i}=\mathcal{W}_{t}^{n,i}-x_{F,t}^{n,i}p_{F,t}-x_{H,t}^{n,i}p_{H,t}$$

	\end{proof}

	\medskip

	\begin{proof}[Proof of Proposition \ref{p: eq_prices}]

Given our price guess, we use the method of undetermined coefficients (MUC) and market clearing to obtain prices. We focus on market clearing for the asset of country i:

	{\small
	\begin{align}
	\phi_{0}^{i}&\frac{E_{t}^{t,i}\left[y_{i,t+1}+p_{i,t+1}\right]-Rp_{i,t}}{\gamma V_{t}^{t,i}\left[y_{i,t+1}+p_{i,t+1}\right]}+\phi_{1}^{i}\frac{E_{t}^{t-1,i}\left[y_{i,t+1}+p_{i,t+1}\right]-Rp_{i,t}}{\gamma V_{t}^{t-1,i}\left[y_{i,t+1}+p_{i,t+1}\right]}+... \\
	&\phi_{1}^{j}\frac{E_{t}^{t,j}\left[y_{i,t+1}+p_{i,t+1}\right]-R p_{i,t}}{\gamma V_{t}^{t,j}\left[y_{i,t+1}+p_{i,t+1}\right]}+\phi_{1}^{j}\frac{E_{t}^{t-1,j}\left[y_{i,t+1}+p_{i,t+1}\right]-R p_{i,t}}{\gamma V_{t}^{t-1,j}\left[y_{i,t+1}+p_{i,t+1}\right]}=1
	\end{align}}

		Plugging in the price guess, after some algebra we obtain:

	{\small
\begin{align}
\phi_{0}^{i}&\left(\frac{\alpha_{i}+\left(1+\beta_{i,0}\right)E_{t}^{t,i}\left[y_{i,t+1}\right]+\beta_{j,0}^{*}E_{t}^{t,i}\left[y_{j,t+1}\right]+\sum_{k=1}^{K}\beta_{i,k}\cdot y_{i,t+1-k}+\sum_{k=1}^{K}\beta_{j,k}^{*}\cdot y_{j,t+1-k}}{\gamma\left[\left(1+\beta_{i,0}\right)^{2}\sigma_{0}^{2}+\beta_{i,0}^{*}{}^{2}\sigma_{0}^{*2}\right]}\right)+... \nonumber \\
\phi_{1}^{i}&\left(\frac{\alpha_{i}+\left(1+\beta_{i,0}\right)E_{t}^{t-1,i}\left[y_{i,t+1}\right]+\beta_{j,0}^{*}E_{t}^{t-1,i}\left[y_{j,t+1}\right]+\sum_{k=1}^{K}\beta_{i,k}\cdot y_{i,t+1-k}+\sum_{k=1}^{K}\beta_{j,k}^{*}\cdot y_{j,t+1-k}}{\gamma\left[\left(1+\beta_{i,0}\right)^{2}\sigma_{1}^{2}+\beta_{i,0}^{*}{}^{2}\sigma_{1}^{*2}\right]}\right)+...	\nonumber \\
\phi_{0}^{j}&\left(\frac{\alpha_{i}+\left(1+\beta_{i,0}\right)E_{t}^{t,j}\left[y_{i,t+1}\right]+\beta_{j,0}^{*}E_{t}^{t,j}\left[y_{j,t+1}\right]+\sum_{k=1}^{K}\beta_{i,k}\cdot y_{i,t+1-k}+\sum_{k=1}^{K}\beta_{j,k}^{*}\cdot y_{j,t+1-k}}{\gamma\left[\left(1+\beta_{i,0}\right)^{2}\sigma_{0}^{*2}+\beta_{i,0}^{*}{}^{2}\sigma_{0}^{2}\right]}\right)+... \nonumber  \\
\phi_{1}^{j}&\left(\frac{\alpha_{i}+\left(1+\beta_{i,0}\right)E_{t}^{t-1,j}\left[y_{i,t+1}\right]+\beta_{j,0}^{*}E_{t}^{t-1,j}\left[y_{j,t+1}\right]+\sum_{k=1}^{K}\beta_{i,k}\cdot y_{i,t+1-k}+\sum_{k=1}^{K}\beta_{j,k}^{*}\cdot y_{j,t+1-k}}{\gamma\left[\left(1+\beta_{i,0}\right)^{2}\sigma_{1}^{*2}+\beta_{i,0}^{*}{}^{2}\sigma_{1}^{2}\right]}\right) \nonumber \\
&\quad \quad \quad=\Sigma\frac{R}{\gamma}\left(\alpha_{i}+\sum_{k=0}^{K}\beta_{i,k}\cdot y_{i,t-k}+\sum_{k=0}^{K}\beta_{j,k}^{*}\cdot y_{j,t-k}\right)
\end{align}}
where
{\footnotesize $$\Sigma\equiv\left(\frac{\phi_{0}^{i}}{\left(1+\beta_{i,0}\right)^{2}\sigma_{0}^{2}+\beta_{i,0}^{*}{}^{2}\sigma_{0}^{*2}}+\frac{\phi_{1}^{i}}{\left(1+\beta_{i,0}\right)^{2}\sigma_{1}^{2}+\beta_{i,0}^{*}{}^{2}\sigma_{1}^{*2}}+\frac{\phi_{0}^{j}}{\left(1+\beta_{i,0}\right)^{2}\sigma_{0}^{*2}+\beta_{i,0}^{*}{}^{2}\sigma_{0}^{2}}+\frac{\phi_{1}^{j}}{\left(1+\beta_{i,0}\right)^{2}\sigma_{1}^{*2}+\beta_{i,0}^{*}{}^{2}\sigma_{1}^{2}}\right).$$	}

\smallskip

We are now ready to use the MUC. We begin by solving for the loadings on the oldest output realizations (where note that $K$ can be an arbitrary large number):

\smallskip

From $y_{j,t-K}$, we have that  $$0=\beta_{j,K}^{*}\cdot y_{j,t-K}\iff\beta_{j,K}^{*}=0$$
and analogously for $y_{i,t-K}$ we have $$0=\beta_{i,K}\cdot y_{i,t-K}\iff\beta_{i,K}=0.$$

\smallskip

From $y_{j,t+1-K}$, and using the previous results, we have
$$0	=\Sigma\frac{R}{\gamma}\beta_{j,K-1}^{*}y_{j,t-(K-1)}\iff\beta_{j,K-1}^{*}=0.$$

Analogously, from $y_{i,t+1-K}$, we have $$0=\Sigma\frac{R}{\gamma}\beta_{i,K-1}y_{i,t-(K-1)}\iff\beta_{i,K-1}=0.$$

Through this process, it is immediate that $\beta_{i,\tau}=\beta_{j,\tau}^{*}=0$ for all $\tau\in\left\{ t-2,K\right]$. We now study the coefficients on the output realizations that also determine beliefs; that is, that have been experienced by at least one market participant. We begin by studying the price loadings on the foreign country output realizations.

\smallskip

From $y_{j,t-1}$, and using the result that $\beta_{2,k}^{*}=0$, we have
{\footnotesize
$$\beta_{j,0}^{*}\left(\frac{\phi_{1}^{i}w_{1}^{*}\left(1-\omega\right)}{\left(1+\beta_{i,0}\right)^{2}\sigma_{1}^{2}+\beta_{i,0}^{*}{}^{2}\sigma_{1}^{*2}}+\frac{\phi_{1}^{j}w_{1}\left(1-\omega\right)}{\left(1+\beta_{i,0}\right)^{2}\sigma_{1}^{*2}+\beta_{i,0}^{*}{}^{2}\sigma_{1}^{2}}\right)=\Sigma R\beta_{j,1}^{*}.$$}

\smallskip

From $y_{j,t}$, after some algebra, we have

{\scriptsize
\begin{align*}
\beta_{j,0}^{*}\left(\Sigma R-\frac{\phi_{0}^{i}w_{1}^{*}}{\left(1+\beta_{i,0}\right)^{2}\sigma_{0}^{2}+\beta_{i,0}^{*}{}^{2}\sigma_{0}^{*2}}-\frac{\phi_{1}^{i}w_{1}^{*}\left(1-\omega\right)}{\left(1+\beta_{i,0}\right)^{2}\sigma_{1}^{2}+\beta_{i,0}^{*}{}^{2}\sigma_{1}^{*2}}-\frac{\phi_{0}^{j}w_{1}}{\left(1+\beta_{i,0}\right)^{2}\sigma_{0}^{*2}+\beta_{i,0}^{*}{}^{2}\sigma_{0}^{2}}-\frac{\phi_{0}^{j}w_{1}\left(1-\omega\right)}{\left(1+\beta_{i,0}\right)^{2}\sigma_{1}^{*2}+\beta_{i,0}^{*}{}^{2}\sigma_{1}^{2}}\right)& \\
=\beta_{1,k}^{*}\Sigma.&
\end{align*}}
Combining the last two conditions, we obtain:
{\scriptsize
$$\beta_{j,0}^{*}=R\frac{\left(\Sigma R-\frac{\phi_{0}^{i}w_{1}^{*}}{\left(1+\beta_{i,0}\right)^{2}\sigma_{0}^{2}+\beta_{i,0}^{*}{}^{2}\sigma_{0}^{*2}}-\frac{\phi_{1}^{i}w_{1}^{*}\left(1-\omega\right)}{\left(1+\beta_{i,0}\right)^{2}\sigma_{1}^{2}+\beta_{i,0}^{*}{}^{2}\sigma_{1}^{*2}}-\frac{\phi_{0}^{j}w_{1}}{\left(1+\beta_{i,0}\right)^{2}\sigma_{0}^{*2}+\beta_{i,0}^{*}{}^{2}\sigma_{0}^{2}}-\frac{\phi_{0}^{j}w_{1}\left(1-\omega\right)}{\left(1+\beta_{i,0}\right)^{2}\sigma_{1}^{*2}+\beta_{i,0}^{*}{}^{2}\sigma_{1}^{2}}\right)}{\left(\frac{\phi_{1}^{i}w_{1}^{*}\left(1-\omega\right)}{\left(1+\beta_{i,0}\right)^{2}\sigma_{1}^{2}+\beta_{i,0}^{*}{}^{2}\sigma_{1}^{*2}}+\frac{\phi_{1}^{j}w_{1}\left(1-\omega\right)}{\left(1+\beta_{i,0}\right)^{2}\sigma_{1}^{*2}+\beta_{i,0}^{*}{}^{2}\sigma_{1}^{2}}\right)}\beta_{j,0}^{*}.$$} 	which generically holds if and only if $$\beta_{j,0}^{*}=0 \Rightarrow \beta_{j,1}^{*}=0.$$

Thus, the price of the asset of country $i$ has zero price loadings on output realizations of country $j$. Finally, we are left with the most recent output realizations of country $i$.

\smallskip

From $y_{i,t-1}$ and using the previous results, we have,

\begin{equation} \label{e: MUC_1}
\left(1+\beta_{i,0}\right)\left(1-\omega\right)\left\{\frac{\phi_{1}^{i}}{\sigma_{1}^{2}}w_{1}+\frac{\phi_{1}^{j}}{\sigma_{1}^{*2}}w_{1}^{*}\right\} =\Sigma R\beta_{i,1}.
\end{equation}

\smallskip

From $y_{i,t}$, we have

\begin{equation} \label{e: MUC_2}
\left(1+\beta_{i,0}\right)\left\{ \frac{\phi_{0}^{i}}{\sigma_{0}^{2}}w_{0}+\frac{\phi_{1}^{i}}{\sigma_{1}^{2}}w_{1}\omega+\frac{\phi_{0}^{j}}{\sigma_{0}^{*2}}w_{0}^{*}+\frac{\phi_{1}^{j}}{\sigma_{1}^{*2}}w_{1}^{*}\omega\right\} +\Sigma\beta_{1}^{i}=\Sigma R\beta_{i,0}.
\end{equation}

\smallskip

And from the constant terms,
{\scriptsize
\begin{equation} \label{e: MUC_3}
\left(1+\beta_{i,0}\right)\left\{ \frac{\phi_{0}^{i}}{\sigma_{0}^{2}}\left(1-w_{0}\right)m_{i}+\frac{\phi_{1}^{i}}{\sigma_{1}^{2}}\left(1-w_{1}\right)m_{i}+\frac{\phi_{0}^{j}}{\sigma_{0}^{*2}}\left(1-w_{0}^{*}\right)m_{i}+\frac{\phi_{1}^{j}}{\sigma_{1}^{*2}}\left(1-w_{1}^{*}\right)m_{i}\right\} =\gamma\left(1+\beta_{i,0}\right)^{2}+\Sigma\left(R-1\right)\alpha_{i}.
\end{equation}}

After some simple algebra and imposing that all cohorts in all countries have a mass of $\frac{1}{4}$, the loadings from the proposition are obtained.

	\end{proof}

\medskip

	\begin{proof}[Proof of Proposition \ref{p: home_bias_priors}]
		To show that a decrease in the prior precision about a foreign country induces a home bias in the domestic country, it suffices to show that $\frac{\partial E[HB_{t}]}{\partial\tau^{*}}|_{\tau^{*}=\tau}<0$, which would imply that $HB_{t}>0$ if $\tau^{*}<\tau$. As $\tau^{*}$ decreases, everything else equal, the aggregate demand in country $H$ of domestic assets, $X_{H,t}^{H}$, remains unchanged. Thus, we focus on the effect of a decrease in the prior precision about the foreign asset on the demand of the domestic country for the foreign asset. We thus highlight the terms in the demand that are directly affected by the prior precision $\tau^*$:

		{\footnotesize
			\begin{align*}
			&X_{F,t}^{H} = \phi_0^H \left[ \frac{\tilde{\alpha}+(1+\beta_{F,0})(1-w^*_{0})m}{\gamma(1+\beta_{F,0})^{2}\sigma_{0}^{*2}} +\frac{w^*_{0}(1+\beta_{F,0})+\beta_{F,1}-R\beta_{F,0}}{\gamma(1+\beta_{F,0})^{2}\sigma_{0}^{*2}} y_{F,t}
			-\frac{R\beta_{F,1}}{\gamma(1+\beta_{F,0})^{2}\sigma_{0}^{*2}} y_{F,t-1} \right] + ...\\
			&\phi_1^{H}\left[\frac{\tilde{\alpha}+\left(1+\beta_{F,0}\right)\left(1-w^*_{1}\right)m}{\gamma\left(1+\beta_{F,0}\right)^{2}\sigma_{1}^{*2}}+\frac{w^*_{1}\left(1+\beta_{F,0}\right)\omega+\beta_{F,1}-R\beta_{F,0}}{\gamma\left(1+\beta_{F,0}\right)^{2}\sigma_{1}^{*2}} y_{F,t} +\frac{w^*_{1}\left(1+\beta_{F,0}\right)\left(1-\omega\right)-R\beta_{F,1}}{\gamma\left(1+\beta_{F,0}\right)^{2}\sigma_{1}^{*2}} y_{F,t-1}\right].
			\end{align*}}

		where the terms indexed by $*$ are a function of $\tau^*$. We can then compute $\frac{\partial X_{F,t}^{H}\left(\tau^{*}\right)}{\partial\tau^{*}}$ as follows
		\begin{equation*}
		=\phi_0^H\frac{\left(y_{F,t}-m\right)\frac{\partial w^*_{0}}{\partial \tau^*}}{\gamma\left(1+\beta_{F,0}\right)\sigma_{0}^{*2}}+\phi_1^{H}\frac{\left(\omega y_{F,t}+\left(1-\omega\right)y_{F,t-1}-m\right)\frac{\partial w^*_{1}}{\partial \tau^*}}{\gamma\left(1+\beta_{F,0}\right)\sigma_{1}^{*2}}+\phi_0^H\frac{\partial x_{F,t}^{t,H}}{\partial\sigma_{0}^{*2}}\frac{\partial\sigma_{0}^{2*}}{\partial\tau^{*}}+\phi_1^{H}\frac{\partial x_{F,t}^{t-1,H}}{\partial\sigma_{1}^{*2}}\frac{\partial\sigma_{1}^{*2}}{\partial\tau^{*}}
		\end{equation*}

		Thus, if $m=E[y_{F,t}]$ for all $t$, it follows that
		\begin{equation*}
		E\left[\frac{\partial X_{F,t}^{H}\left(\tau^{*}\right)}{\partial\tau^{*}}\right] = \phi_0^H\frac{\partial x_{F,t}^{t,H}}{\partial\sigma_{0}^{*2}}\frac{\partial\sigma_{0}^{2*}}{\partial\tau^{*}}+\phi_1^{H}\frac{\partial x_{F,t}^{t-1,H}}{\partial\sigma_{1}^{*2}}\frac{\partial\sigma_{1}^{*2}}{\partial\tau^{*}} < 0
		\end{equation*}

Thus, everything else equal, as the precision of prior beliefs about the output in the foreign country decreases, domestic  agents demand less of the foreign asset. In equilibrium, prices fall to off-set some of this effect, but since the asset supply is fixed, it follows that in the equilibrium with $\tau^*<\tau$, and prior beliefs centered around the truth, domestic agents hold on average a smaller fraction of the foreign asset than of the domestic asset, $E[X_H^H]>E[X_F^H]$.

\end{proof}

\begin{proof}[Proof of Proposition \ref{p: priors_recessions}]
We define flows for country $i$ as the change in the domestic asset holdings of domestic and foreign investors, respectively:
$$\Delta_{i,t}^{i}	=\left(X_{i,t}^{i}-X_{i,t-1}^{i}\right).$$
$$\Delta_{i,t}^{j}	=\left(X_{i,t}^{j}-X_{i,t-1}^{j}\right).$$


It is immediate that $X_{i,t-1}^{i}$ and $X_{i,t-1}^{j}$ are independent of $y_{i,t}$, as they are portfolios holdings chosen in $t-1$. As a result, to understand the change in flows, it suffices to check that:
$$\frac{\partial X^i_{i,t}}{\partial y_{i,t}}<0 \quad \quad \quad \frac{\partial X^j_{i,t}}{\partial y_{i,t}}>0.$$

In equilibrium, $X_{i,t}^{i}+X_{i,t}^{j}=1$, which implies that $\frac{\partial X_{i,t}^{i}}{\partial y_{i,t}}	=-\frac{\partial X_{i,t}^{j}}{\partial y_{i,t}}.$ Thus, it suffices to show that $\frac{\partial X_{i,t}^{i}}{\partial y_{i,t}}<0.$

\begin{align*}
\frac{\partial X^i_{i,t}}{\partial y_{i,t}} = \phi_0^i \left(\frac{w_{0}\left(1+\beta_{0,i}\right)+\beta_{1,i}-R\beta_{0,i}}{\gamma\left(1+\beta_{0,i}\right)^{2}\sigma_{0}^{2}}\right) + \phi_1^i \left(\frac{w_{1}\left(1+\beta_{0,i}\right)\omega+\beta_{1,i}-R\beta_{0,i}}{\gamma\left(1+\beta_{0,i}\right)^{2}\sigma_{1}^{2}}\right)
\end{align*}

by replacing $\beta_{1,i}-R\beta_{0,i}$ from using equation \eqref{e: MUC_2} and using the fact that and using the fact $1+\beta_{0,i}>0$, we  obtain

{\small
\begin{align} \label{e: changes_flows}
\frac{\partial X_{i,t}^{i}}{\partial y_{i,t}}&\propto\left(\frac{\phi_{0}^{i}}{\sigma_{0}^{2}}w_{0}+\frac{\phi_{1}^{i}}{\sigma_{1}^{2}}w_{1}\omega\right)-\left(\frac{\phi_{0}^{i}}{\sigma_{0}^{2}}+\frac{\phi_{1}^{i}}{\sigma_{1}^{2}}\right)\frac{\left(\frac{\phi_{0}^{i}}{\sigma_{0}^{2}}w_{0}+\frac{\phi_{1}^{i}}{\sigma_{1}^{2}}w_{1}\omega+\frac{\phi_{0}^{j}}{\sigma_{0}^{*2}}w_{0}^{*}+\frac{\phi_{1}^{j}}{\sigma_{1}^{*2}}w_{1}^{*}\omega\right)}{\left(\frac{\phi_{0}^{i}}{\sigma_{0}^{2}}+\frac{\phi_{1}^{i}}{\sigma_{1}^{2}}+\frac{\phi_{0}^{j}}{\sigma_{0}^{*2}}+\frac{\phi_{1}^{j}}{\sigma_{1}^{*2}}\right)} \\
&=\frac{\phi_{0}^{i}}{\sigma_{0}^{2}}\frac{\phi_{0}^{j}}{\sigma_{0}^{*2}}\left(w_{0}-w_{0}^{*}\right)+\frac{\phi_{0}^{i}}{\sigma_{0}^{2}}\frac{\phi_{1}^{j}}{\sigma_{1}^{*2}}\left(w_{0}-w_{1}^{*}\omega\right)+\frac{\phi_{0}^{j}}{\sigma_{0}^{*2}}\frac{\phi_{1}^{i}}{\sigma_{1}^{2}}\left(w_{1}\omega-w_{0}^{*}\right)+\left(\frac{\phi_{1}^{i}}{\sigma_{1}^{2}}\frac{\phi_{1}^{j}}{\sigma_{1}^{*2}}\right)\left(w_{1}-w_{1}^{*}\right)\omega \nonumber
\end{align}}

Now, when both countries have the same demographics of market participation, $\phi_{age}^{i}=\phi$ for all $i\in\left\{ H,F\right\}$  and $age\in\left\{ 0,1\right\}$, where $\phi=\frac{1}{4}$ in the baseline model. With this,
{\small
$$\frac{\partial X_{i,t}^{i}}{\partial y_{i,t}}=\frac{\phi^{2}}{\sigma_{0}^{2}\sigma_{0}^{*2}}\left[\left(1+\frac{\sigma_{0}^{*2}}{\sigma_{1}^{*2}}\right)w_{0}-\left(1+\frac{\sigma_{0}^{2}}{\sigma_{1}^{2}}\right)w_{0}^{*}\right]+\frac{\phi^{2}}{\sigma_{1}^{2}\sigma_{1}^{*2}}\left[\left(1+\frac{\sigma_{1}^{*2}}{\sigma_{0}^{*2}}\right)w_{1}-\left(1+\frac{\sigma_{1}^{2}}{\sigma_{0}^{2}}\right)w_{1}^{*}\right]\omega.$$}

Remember that
\begin{align*}
\sigma_{0}^{2}&=\left[\tau+\sigma^{-2}\right]^{-1} \quad  \quad \quad \;\;\; \sigma_{1}^{2}=\left[\tau+2\sigma^{-2}\right]^{-1} \\
\sigma_{0}^{*2}&=\left[\tau^{*}+\sigma^{-2}\right]^{-1} \quad  \quad \quad \sigma_{1}^{*2}=\left[\tau^{*}+2\sigma^{-2}\right]^{-1} \\
w_{0}&=\frac{\sigma^{-2}}{\tau+\sigma^{-2}}\quad  \quad \quad\quad\quad \;\; w_{1}\omega=\frac{\sigma^{-2}}{\tau+2\sigma^{-2}} \\
w_{0}^{*}&=\frac{\sigma^{-2}}{\tau^{*}+\sigma^{-2}}\quad  \quad \quad\quad \quad w_{1}^{*}\omega=\frac{\sigma^{-2}}{\tau^{*}+2\sigma^{-2}}
\end{align*}

Plugging this in, we obtain that the first and second term in brackets are:

\begin{align}
\left(1+\frac{\sigma_{0}^{*2}}{\sigma_{1}^{*2}}\right)w_{0}-\left(1+\frac{\sigma_{0}^{2}}{\sigma_{1}^{2}}\right)w_{0}^{*} &=\frac{2\sigma^{-2}}{\left(\tau+\sigma^{-2}\right)\left(\tau^{*}+\sigma^{-2}\right)}\left(\tau^{*}-\tau\right)<0 \\
\left(1+\frac{\sigma_{1}^{*2}}{\sigma_{0}^{*2}}\right)w_{1}\omega-\left(1+\frac{\sigma_{1}^{2}}{\sigma_{0}^{2}}\right)w_{1}^{*}\omega &=\frac{2\sigma^{-2}}{\left(\tau+2\sigma^{-2}\right)\left(\tau^{*}+2\sigma^{-2}\right)}\left(\tau^{*}-\tau\right)<0
\end{align}
Thus, it follows that with same demographics, $\frac{\partial X_{i,t}^{i}}{\partial y_{i,t}}<0$ and $\frac{\partial X_{i,t}^{j}}{\partial y_{i,t}}>0$.

%

\end{proof}

\begin{proof}[Proof of Proposition \ref{p: demographics}]
	To analyze the effect of different demographics, we analyze how the change in flows varies with changes in the mass of young and old market participants in both the domestic and the foreign country:

	\begin{align}
	\frac{\partial}{\partial\phi_{0}^{j}}\left(\frac{\partial X_{i,t}^{i}}{\partial y_{i,t}}\right)	&=\frac{\phi_{0}^{i}}{\sigma_{0}^{2}\sigma_{0}^{*2}}\left(w_{0}-w_{0}^{*}\right)+\frac{\phi_{1}^{i}}{\sigma_{0}^{*2}\sigma_{1}^{2}}\left(w_{1}\omega-w_{0}^{*}\right) \label{e: der_1}\\
	\frac{\partial}{\partial\phi_{1}^{j}}\left(\frac{\partial X_{i,t}^{i}}{\partial y_{i,t}}\right)	&=\frac{\phi_{0}^{i}}{\sigma_{0}^{2}\sigma_{1}^{*2}}\left(w_{0}-w_{1}^{*}\omega\right)+\frac{\phi_{1}^{i}}{\sigma_{1}^{2}\sigma_{1}^{*2}}\left(w_{1}-w_{1}^{*}\right)\omega \label{e: der_2} \\
	\frac{\partial}{\partial\phi_{0}^{i}}\left(\frac{\partial X_{i,t}^{i}}{\partial y_{i,t}}\right)	&=\frac{\phi_{0}^{j}}{\sigma_{0}^{2}\sigma_{0}^{*2}}\left(w_{0}-w_{0}^{*}\right)+\frac{\phi_{1}^{j}}{\sigma_{0}^{2}\sigma_{1}^{*2}}\left(w_{0}-w_{1}^{*}\omega\right) \label{e: der_3} \\
	\frac{\partial}{\partial\phi_{1}^{i}}\left(\frac{\partial X_{i,t}^{i}}{\partial y_{i,t}}\right)	&=\frac{\phi_{0}^{j}}{\sigma_{0}^{*2}\sigma_{1}^{2}}\left(w_{1}\omega-w_{0}^{*}\right)+\frac{\phi_{1}^{j}}{\sigma_{1}^{2}\sigma_{1}^{*2}}\left(w_{1}-w_{1}^{*}\right)\omega \label{e: der_4}
	\end{align}

	First, when $\tau\geq \tau^*$, we have that $w_0\leq w_0^*$, $w_1\leq w_1^*$, and $w_1\omega < w_1^*\omega \leq w_0^*$, as foreign (and young) agents overweight domestic output realizations relative to domestic (and old) agents. As a result, it is immediate that the sign of \eqref{e: der_1} and \eqref{e: der_4} is negative.

	\medskip

	Next, we analyze the sign of \eqref{e: der_2} and \eqref{e: der_3}. Note that for $\tau=\tau^{*}$, as $w_{age}=w_{age}^*$ for $age\in\{0,1\}$,  we have
	\begin{align*}
	\frac{\partial}{\partial\phi_{1}^{j}}\left(\frac{\partial X_{i,t}^{i}}{\partial y_{i,t}}\right)	&=\frac{\phi_{0}^{i}}{\sigma_{0}^{2}\sigma_{1}^{2}}\left(w_{0}-w_{1}\omega\right)>0 \\
	\frac{\partial}{\partial\phi_{0}^{i}}\left(\frac{\partial X_{i,t}^{i}}{\partial y_{i,t}}\right)	&=\frac{\phi_{1}^{j}}{\sigma_{0}^{2}\sigma_{1}^{2}}\left(w_{0}-w_{1}\omega\right)>0
	\end{align*}
	and that for for $\tau>\tau^{*}+\sigma^{-2}$, as  $w_{0}=\frac{\sigma^{-2}}{\tau+\sigma^{-2}}<\frac{\sigma^{-2}}{\tau^{*}+2\sigma^{-2}}=w_{1}^{*}\omega$, we have
	\begin{align*}
	\frac{\partial}{\partial\phi_{1}^{j}}\left(\frac{\partial X_{i,t}^{i}}{\partial y_{i,t}}\right)&<0     \\ \frac{\partial}{\partial\phi_{0}^{i}}\left(\frac{\partial X_{i,t}^{i}}{\partial y_{i,t}}\right)&<0.
	\end{align*}
	Finally, by plugging in the explicit formulas for variances and weights that depend on $\tau$, it is straightforward that these sensitivities are monotonically decreasing in $\tau$:

	\begin{align*}
	\frac{\partial}{\partial\phi_{1}^{j}}\left(\frac{\partial X_{i,t}^{i}}{\partial y_{i,t}}\right)	&
	=\frac{\phi_{0}^{i}}{\sigma_{1}^{*2}}\left(\sigma^{-2}-\left(\tau+\sigma^{-2}\right)\omega_{1}\omega\right)+\frac{\phi_{1}^{i}}{\sigma_{1}^{*2}}\left(\sigma^{-2}-\left(\tau+2\sigma^{-2}\right)w_{1}^{*}\omega\right)\\
	\frac{\partial}{\partial\phi_{0}^{i}}\left(\frac{\partial X_{i,t}^{i}}{\partial y_{i,t}}\right)	
	&=\frac{\phi_{0}^{j}}{\sigma_{0}^{*2}}\left(\sigma^{-2}-\left(\tau+\sigma^{-2}\right)w_{0}^{*}\right)+\frac{\phi_{1}^{j}}{\sigma_{1}^{*2}}\left(\sigma^{-2}-\left(\tau+\sigma^{-2}\right)w_{1}^{*}\omega\right)
	\end{align*}

	Thus, it follows that there exists a $\bar{\tau}_{1},\bar{\tau}_{2}\in\left(\tau^{*},\tau^{*}+\sigma^{-2}\right)$ such that $\frac{\partial}{\partial\phi_{1}^{j}}\left(\frac{\partial X_{i,t}^{i}}{\partial y_{i,t}}\right)>0$ if and only if $\tau<\bar{\tau}_{1}$ and $\frac{\partial}{\partial\phi_{0}^{i}}\left(\frac{\partial X_{i,t}^{i}}{\partial y_{i,t}}\right)>0$ if and only if $\tau<\bar{\tau}_{2}$.
\end{proof}

\section{Data}\label{app:Data}

This appendix provides a more detailed description of our data and variable construction, and complements the explanations in the main text.

Our main sources of data are the International Monetary Fund's (IMF) International Financial Statistics (IFS), including the International Investment Positions (IIP), which is available from \url{http://data.imf.org/IFS}; 
Balance of Payments (BOP) data, also provided by the IMF;
the World Bank's World Development Indicators (WDI), available at \url{https://datacatalog.worldbank.org/dataset/world-development-indicators};
and the United Nations' (UN) data on World Population Prospects (WPP), which we obtain from \url{https://population.un.org/wpp/Download/Standard/Population/} under Annual Population by Age Groups - Both Sexes.
As there is no official crosswalk for the UN, IMF, and World Bank data, we generate a country crosswalk by hand to merge the different data sources. 

As discussed in the main text, our key variables are yearly measures of Home Bias in the US, capital outflows by domestic agents (COD), capital inflows by foreign agents (CIF), and the growth of the Gross 
Domestic Product (GDP growth), all of which we defined and discuss in detail below. 
Table \ref{tab:summary} summarizes the sample sizes and variable characteristics. All monetary data are in US\$ and inflation adjusted using the GDP deflator. 

\begin{table}[htbp!]
	\caption{Summary Statistics}
	\vspace{-0.4cm}
	\begin{small}
		\begin{tablenotes}
			\item Home Bias based on data from 1980-2017. For COD, CIF, and GDP growth year panel varies by country. 
		\end{tablenotes}
	\end{small}
	\vspace{0.5em}
	\begin{center}
		{
			\def\sym#1{\ifmmode^{#1}\else\(^{#1}\)\fi}
			\begin{tabular}{l*{5}{c}}
				\hline\hline
				Variable & N  &  Mean  & SD  & Min  &  Max \Tstrut \Bstrut \\ \hline
				Home Bias & 38  & 0.44   & 0.09& 0.32 & 0.65  \Tstrut\Bstrut\\\hline
				COD  & 725 & 0 &  .99 &  -3.02  & 3.31 \Tstrut\\
				CIF & 714 & 0 &  .99 & -2.74   & 4.07\\
				GDP Growth & 874 & 0.05 & 0.15 & -0.64 & 1.79 \Bstrut \\
				\hline\hline
			\end{tabular}
		}
	\end{center}
	\label{tab:summary}
\end{table}

\subsection{Home Bias}
\label{appsec:homebias}

The construction of our measure of Home Bias is detailed in the main text, and summarized in the formula$$\frac{\textrm{Domestic Equity Investment in the US}}{\textrm{Global Equity Investment by the US}} - \frac{\textrm{US Market Capitalization}}{\textrm{Global Market Capitalization}}.$$ 

As discussed in Section \ref{ss:HomeBias}, we obtain data on domestic and foreign equity investment in the US and globally from the IIP data of the IMF, and capitalization data from the World Federation of Exchanges. Here, we provide more details on the location of the data and the variable construction.

For equity investment, we follow \cite{cooper2013equity}, \cite{mishra2015measures}, and \cite{hau2008home} and use `equity portfolio investment' as defined by the IMF in their IIP data and available for 1976-2017. 
(Note that the start and end years in the IIP data vary by country, some countries have gaps.) This class of investments captures liquid investment assets such as listed equities, and also includes investments of less than 10\% in venture capital firms, while investments above 10\% in a listed company are excluded (and counted as direct investment). The restriction to listed equity in our (and in prior) research reflects that robust data is generally not available for unlisted equities.
For each country and sample year, the IIP provides (1) a country's assets, which is equal to the total foreign equity portfolio investment holdings for a country in that year, and (2) a country's liabilities, which is equal to the total equity portfolio investment holdings by foreigners in a country in that year. 

We calculate domestic equity investment by the US by subtracting foreign investment in listed US equity from the US listed-equity market capitalization.
Foreign investment in the US comes from the IIP data within the IFS, specifically ``International Investment Position, Liabilities, Portfolio investment, Equity and investment fund shares [BPM6], US Dollar (ILPE\_BP6\_USD).'' US Market Capitalization comes from the World Federation of Exchanges, available from the World Bank's WDI data for 1980-2017\footnote{The Center for Research in Security Prices (CRSP) provides more historical market capitalization for the US, but not for the world.}, specifically ``Market capitalization of listed domestic companies (current US\$) [CM.MKT.LCAP.CD].''
To calculate total equity investment by the US we sum the US listed-equity investment by the US and foreign listed-equity investment by the US, also provided directly by the IIP imder ``International Investment Position, Assets, Portfolio investment, Equity and investment fund shares [BPM6], US Dollar (IAPE\_BP6\_USD).''

For the benchmark measure, the percent of the global portfolio allocated to the US, we divide the US market capitalization by the global market capitalization, both available from the World Federation of Exchanges (via the World Bank's WDI).

Finally, we note that some of the prior research on home bias, such as \cite{cooper2013equity} and \cite{mishra2015measures}, use the Coordinated Portfolio Investment Survey (CPIS) from the IMF which provides more granular data than the IIP. However, the CPIS starts only in 2001, which would halve our sample size. The only potential advantage to using the CPIS over the IIP data is the ability to estimate ``round-tripping,'' when money is invested domestically but through a tax haven so it is counted as a foreign investment in the IIP data (cf.\;\citealp{cooper2013equity}). However, the CPIS data only allows for a rather imprecise estimation of ``round-tripping'' which could bias results in an unknown direction. By not accounting for ``round-tripping,'' our measure is biased towards less domestic US investment and hence
less home bias. As a result, our results are conservative.

\subsection{Fickleness and Retrenchment}

As indicated in the main text, we calculate our measures of capital flows from annual changes in assets and liabilities in the IMF's BOP data, following the variable definitions of capital inflows by foreign agents (CIF) and capital outflows by domestic agents (COD) in \cite{caballero2016model}.

For CIF, we calculate the sum of changes in a country's liabilities from direct investment, portfolio investment, and  other investment using the BOP data series ``Financial Account, Net Lending ($+$) / Net Borrowing ($-$) (Balance from Financial Account), Direct Investment, Net Incurrence of Liabilities, US Dollars (BFDL\_BP6\_USD),'' ``Financial Account, Portfolio Investment, Net Incurrence of Liabilities, US Dollars (BFPL\_BP6\_USD),'' and ``Financial Account, Other Investment, Net Incurrence of Liabilities, US Dollars (BFOL\_BP6\_USD).''

For COD, we calculate the sum of changes in a country's assets from direct investment, portfolio investment, other investment, and reserve assets from the financial account using the BOP data series ``Financial Account, Net Lending ($+$) / Net Borrowing ($-$) (Balance from Financial Account), Direct Investment, Net Acquisition of Financial Assets, US Dollars (BFDA\_BP6\_USD),'' ``Financial Account, Portfolio Investment, Net Acquisition of Financial Assets, US Dollars (BFPA\_BP6\_USD),'' ``Financial Account, Other Investment, Net Acquisition of Financial Assets, US Dollars (BFOA\_BP6\_USD)'', and ``Financial Account, Reserve Assets, US Dollars (BFRA\_BP6\_USD).''

We inflation-adjust, normalize, and standardize the flow variables, as discussed in the main text.
For the inflation adjustment, we use the GDP deflator data available from the St. Louis Federal Reserve Bank here at  \url{https://fred.stlouisfed.org/series/GDPDEF}. (We adjust the series to 100 in 2010.)
To normalize CIF and COD, we divide by trend GDP, given by the Hodrick-Prescott filter of \cite{hodrick1997postwar} using the recommended smoothing parameter of \cite{ravn2002adjusting} of 6.25. After standardizing the variables to a mean of 0 and a standard deviation of 1 (at the country level), the summary statistic of standard deviations in Table \ref{tab:summary} is slightly below 1 due to numerical error; however, when tabulated, each individual country has a standard deviation of 1 and a mean of 0 up to a high number of digits.

Real GDP growth is calculated from yearly GDP data in contemporaneous US dollars provided by the World Bank in the World Development Indicators from 1970-2017 (varying by country), as discussed in the main text.
The year of maximum GDP growth, from 1973 to 1974 in Saudi Arabia, resulted from the rapid increase of oil sales in the Middle East in the 1970s. The year of minimum GDP growth is from 2001 to 2002 in Argentina, resulting from the Argentinian debt crisis of the early 2000s. Since these are not anomalous outliers, but driven by  major events, we follow \cite{broner2013gross} by opting not to drop or winsorize any data points.

For the population distribution, we use yearly, country-level population data from the UN which gives population by country in 5-year age bins estimates from 1950-2020.

We use four country level population-based indicator variables. We follow \cite{malmendierpouzovanasco2018}, in particular the definitions used for calculations using a Markov-Switching Regime, who define the investor population as those age 25 to 75 and define young as those below 50. Working within the 5-year age bins provided by the UN, we define the young investor population of a country as those aged 25 to 49 and the old investor population as those aged 50 to 74. 
We calculate all indicators based on countries' relative positions within our panel of G20 countries.
Since population is a smooth variable and, relative to other countries, similar across time, we calculate country-level, as opposed to a country-year-level, indicators. 
To do this we average each country's young and old population from 1970 to 2018. We then determine which countries have above median and top-quartile levels of average young and old population. Even though size is an important determinant of whether or not a country has a high level of young/old investors, it is not the sole determinant and our variables are not collinear.

\end{appendices}
\end{document}